\documentclass[sn-mathphys]{sn-jnl}

\jyear{2024}%

\theoremstyle{thmstyleone}%
\newtheorem{theorem}{Theorem}
\usepackage{amsmath}
\usepackage{amssymb}
\usepackage{amsfonts}
\usepackage{amsthm}
\usepackage{graphicx}
\usepackage{physics}
\usepackage{array}
\usepackage[capitalize]{cleveref}

\newcommand{\mbf}{\mathbf}

\newcommand{\m}[1]{\mbf{#1}}

\newcommand{\jac}{\m J}

\newcommand{\beq}{\begin{equation}\begin{aligned}}
\newcommand{\eeq}{\end{aligned}\end{equation}}

\newcommand{\D}{\mbf{\hat D}}
\renewcommand{\u}{\mbf u}
\newcommand{\Db}{\mathbb D}

\newcommand{\linop}{\jac - \mu  \D}

\newcommand{\ak}[1]{#1}

\newtheorem{lemma}{Lemma}

\DeclareMathOperator{\diag}{diag}

\theoremstyle{thmstyletwo}%
\newtheorem{remark}{Remark}%

\theoremstyle{thmstylethree}%

\begin{document}

\title[Designing instabilities of reaction-cross-diffusion systems]{Designing reaction-cross-diffusion systems with Turing and wave instabilities}


\author[1]{\fnm{Edgardo} \sur{Villar-Sep\'ulveda}}\email{edgardo.villar-sepulveda@bristol.ac.uk}

\author[1]{\fnm{Alan R.} \sur{Champneys}}\email{a.r.champneys@bristol.ac.uk}

\author[2]{\fnm{Andrew L.} \sur{Krause}}\email{andrew.krause@durham.ac.uk}


\affil[1]{{Engineering Mathematics}, {University of Bristol}, {{Ada Lovelace Building, Tankard's Cl, University Walk}, {Bristol}, {BS8 1TW}, United Kingdom }}
\affil[2]{Mathematical Sciences Department, Durham University,
Upper Mountjoy Campus, Stockton Rd,
Durham DH1 3LE, United Kingdom }




\abstract{General conditions are established under which reaction-cross-diffusion systems can undergo spatiotemporal pattern-forming instabilities. Recent work has focused on designing systems theoretically and experimentally to exhibit patterns with specific features, but the case of non-diagonal diffusion matrices has yet to be analysed. Here, a framework is presented for the design of general $n$-component reaction-cross-diffusion systems that exhibit Turing and wave instabilities of a given wavelength. For a fixed set of reaction kinetics, it is shown how to choose diffusion matrices that produce each instability; conversely, for a given diffusion tensor, how to choose linearised kinetics. The theory is applied to several examples including a hyperbolic reaction-diffusion system, two different 3-component models, and a spatio-temporal version of the Ross-Macdonald model for the spread of malaria.}

\keywords{Reaction-diffusion; Diffusion-driven instability; Spatiotemporal oscillations; Turing instability; wave instability}


\pacs[MSC Classification]{35K57, 35B36}

\maketitle

\section{Introduction}\label{sec1}
    Spontaneous pattern formation in spatially extended systems has been widely studied following Turing's 1952 paper \cite{Turing1952} and remains an area of substantial current work both theoretically \cite{cross1993pattern, Anotida,FahadWoods,villar-sepulveda,Avitabile,paquin-lefebvre, krause2021modern, al2021unified}, and in a variety of applied contexts \cite{kondo2010reaction, maini2012turing, painter2021systems, ramos2024parsing}. Substantial emphasis has been placed on 2-component reaction-diffusion equations, as these are the first and one of the simplest models of pattern formation \cite{Murray2001}. Extensions to three and more component systems have also been carried out, largely using linear stability theory and direct numerical simulations \cite{Yang2,yochelis-knobloch,anma2012unstable,zheng2016identifying, kuznetsov2022robust}. Importantly, there has been substantially less work on systems with more than three species despite their ubiquity in many areas of application (e.g.~gene regulatory networks responsible for key aspects of embryological development can include thousands of genes \cite{johnson2009functional}); see \cite{satnoianu, marcon2016high, scholes2019comprehensive} for examples of theoretical studies of $n$-component reaction-diffusion systems.
    
    In recent work, the first two authors developed general conditions under which an $n$-component reaction-diffusion system can undergo Turing or wave instabilities \cite{villar-sepulveda}, generalizing ideas from the 3-component case \cite{anma2012unstable}. Given the linearized reaction kinetics of the system, the authors determined if there exists a diffusion tensor such that the system can admit Turing or wave instabilities, and if so provided a procedure for constructing a diagonal diffusion tensor for which these instabilities of specified wavelengths would emerge. These ideas are, in a sense, a generalization of the large-domain approximation used in the two-component case to determine `Turing conditions' on pattern formation which are a function only of kinetic parameters and diffusion rates (see e.g.~\cite[Chapter 2]{Murray2001}).

    However, studying the stability of more general reaction-transport systems will not necessarily lead to diagonal diffusion tensors after linearization. Instead, it may give rise to a range of more complex linear operators. One well-studied class of such systems comprises reaction-cross-diffusion systems involving fluxes of some components influencing the growth rates of others. Such systems are becoming increasingly well-studied, particularly since the pioneering work of Keller and Segel \cite{Keller-Segel} as well as Shigesada, Kawasaki, and Teramoto \cite{shigesada1979spatial}; see \cite{lou, epstein-cross, gambino2012turing, breden2019influence, taylor2020non, ritchie2022turing, gaffney2023spatial} for some recent examples of reaction-cross-diffusion systems. Most of this work focuses again on the two or three-component cases, but there are no general results, to our knowledge, on linear instabilities for $n$-component systems (along the same lines as \cite{villar-sepulveda,satnoianu} for reaction-diffusion equations).

    Part of the reason for this is that reaction-cross-diffusion systems are not as restrictive as classical reaction-diffusion systems in terms of conditions for pattern formation; for example, cross-diffusion terms can induce Turing instabilities in systems with two self-inhibiting components, even if both components' self-diffusion coefficients are equal. Another reason is that reaction-diffusion systems with cross-diffusion can be mathematically ill-posed in several cases, leading to finite-time singularities; see e.g.~\cite{lankeit2020facing}. At the fundamental level of diffusion processes, the thermodynamic foundation for such non-diagonal diffusion tensors is also problematic, see e.g.~\cite{mendez2010reaction, klika2018beyond}. Nevertheless, many reaction-cross-diffusion systems arise as models in applied mathematics, and our goal here is to generalize the notion of Turing conditions for pattern formation to incorporate cross-diffusion. In particular, we want to study when such a system can admit Turing and/or wave instabilities. Within that mathematical modelling setting, we ask how one could design diffusion matrices or linearized reaction kinetics to admit such instabilities. Our objective is to give conditions under which designing unstable systems is possible, alongside particular approaches to doing so.

    Specifically, we shall consider either fixing the linearized kinetics of the system and varying the diffusion tensor or fixing the diffusion tensor and varying the linearized kinetics. Assuming stability to homogeneous perturbations (i.e.~that the system's linearized kinetics are stable in the absence of transport), we then ask how to design the remaining free matrix (the diffusion tensor or the linearized kinetics) so that the resulting reaction-diffusion system becomes unstable around its steady state after the addition of transport; that is so that the system undergoes a Turing or a wave instability. In the end, this can be seen as a matter of designing a reaction-diffusion system, subject to constrained transport or constrained kinetics, so that it generates diffusion-driven instabilities under the presence of cross-diffusion. Through examples, we will also explore situations where one does not have complete freedom over either of these matrices and must design the system more carefully. Such questions of design principles are becoming increasingly important both in synthetic biology \cite{diambra2015cooperativity,karig2018stochastic, santos2019using}, as well as in materials science and other areas \cite{tanaka2023turing,luo2024wrinkled}. Previous theoretical work on such design principles has largely focused on reaction-diffusion systems of 2 components \cite{vittadello2021turing,woolley2021bespoke,leyshon2021design}, so the key novelty in this paper is to consider $n$-component reaction-cross-diffusion systems as well as to allow for designing wave instabilities, in addition to Turing instabilities.

   \subsection{Statement of the main results}
        For simplicity of presentation, we will focus on reaction-cross-diffusion systems on the real line, hence not being concerned with boundary conditions or geometry, as long as the eigenvalues of the Laplacian on the domain are sufficiently well-behaved. We note that this is essentially without loss of generality \ak{in terms of problems posed on Riemannian manifolds with Neumann or periodic boundary conditions}; see \cite{Murray2001, krause2021modern} for discussion of this point in general, and \cite[Remark 1]{villar-sepulveda} in particular. We do note that the use of Neumann or periodic boundary conditions can have important consequences for wave instabilities as these can give rise to different kinds of patterns, \textit{travelling waves} or \textit{standing waves} \cite{villar-sepulveda,Villar-sepulveda-wave,knobloch}, \ak{and that other boundary conditions can induce a variety of behaviours \cite{maini1997boundary, krause2021isolating}}. Given a fixed geometry and boundary conditions, one can simply scale the free matrix discussed below to fit a domain of a particular geometry, in any spatial dimension, assuming the Laplace (or more generally, the Laplace-Beltrami operator) spectrum is known and well-behaved \ak{(i.e.~consisting solely of point eigenvalues)} for the given boundary conditions.
    
        We pose a general reaction-cross-diffusion system as:
        \begin{align}
            \partial_t \u = \mbf f(\u) + \partial_x \left(\m{D}(\u)  \partial_x \u\right), \label{general_nonlinear_system}
        \end{align}
        where $\u \in \mathbb R^n$, $\mbf f: \mathbb R^n \to \mathbb R^n$, and $\m{D}: \mathbb R^n \to \mathbb R^n \times \mathbb R^n$ is the diffusion tensor of the system, which is not necessarily diagonal. We will assume throughout that $n \geq 2$, though note that $n = 2$ does not allow for wave instabilities even with cross-diffusion \cite{ritchie2022turing}. We assume that sufficiently smooth and unique solutions exist, noting that the existence and regularity theory for these systems is much more intricate than for simpler reaction-diffusion models, with blowup and singularities having a significant literature \cite{choi2004existence,le2005regularity, seis-well,choquet-well}; see \cite{lankeit2020facing} for an introductory review to these complexities and their analysis. We will henceforth neglect details regarding well-posedness and tacitly assume that everything at the nonlinear level is well-behaved.
        
        We assume that \cref{general_nonlinear_system} admits a spatially homogeneous steady state $\mbf P$ such that $\mbf f(\mbf P) = \mbf 0$, \ak{and that $\mbf D(\mbf P)$ has strictly positive real eigenvalues, and hence is invertible for all the values of its argument}. After linearizing  about this steady state by writing $\u = \mbf P + \varepsilon \mbf u_1$ and keeping terms up to $\mathcal O(\varepsilon)$, we 
        obtain
        \beq
            \partial_t \mbf{u_1} =  \jac \mbf{u_1} + \D  \partial_{xx} \mbf {u_1}, \label{generalsystem}
        \eeq
        where the constant matrices
        $$
            \D = \mbf  D(\mbf P), \qquad \mbox{and} \quad \jac = \partial_{\mbf u} \mbf f (\mbf P),
        $$
        correspond to diffusion and Jacobian matrices, respectively. Henceforth, we will work with these two constant matrices, dropping the explicit dependence on the steady-state $\mbf P$.

        We proceed to solve \cref{generalsystem} in the usual way using separable solutions composed of Laplacian eigenfunctions and solutions in time of the form $\exp(\lambda t)$. This leads to solvability conditions determining growth rates $\lambda$ as solutions to
        \beq \label{general_solvability_condition}
            \det\left(\linop - \lambda \m{I}\right) = 0,
        \eeq
        where $\mu = k^2$ is an eigenvalue of the negative Laplacian, and $k$ is a wavenumber. Instabilities of the full system, \cref{general_nonlinear_system}, correspond to solutions of \cref{general_solvability_condition} with $\Re(\lambda) > 0$. Henceforth, we will assume that $\jac$ is a stable matrix; that is, all of its eigenvalues have strictly negative real parts.

        In the classical two-species case, we need the Jacobian to have the right sign structure, with two positive and two negative entries \cite{Murray2001}. Here, we only assume that it is a non-scalar matrix (i.e.~not a scalar multiple of the identity matrix).
        We then say that the system given by \cref{generalsystem} exhibits a \emph{Turing instability} if there is at least one purely real value of $\lambda > 0$ satisfying \cref{general_solvability_condition}. Also, we say that the system exhibits a \emph{wave instability} if there is at least one complex conjugate pair of $\lambda$ with positive real part and nontrivial imaginary part satisfying \cref{general_solvability_condition}. In both cases, we can then say that the full nonlinear system \cref{general_nonlinear_system} exhibits these instabilities around the steady state $\mbf P$.

        We now state our main results. We consider $\ell$ to be a non-negative integer in what comes below.
        \begin{theorem} \label{th:main}
            Let $\jac$ be a non-scalar stable matrix. Then, there exists a diagonalizable diffusion matrix $\D$ such that \cref{generalsystem} exhibits a Turing instability. Furthermore, if $n \geq 3$ and $\jac$ has at least $1 \leq \ell \leq 3$ distinct eigenvalues with at least $3 - \ell$ corresponding generalized eigenvectors, then there is also a diagonalizable diffusion matrix $\D$ such that \cref{generalsystem} exhibits a wave instability.
        \end{theorem}
        \begin{theorem} \label{th:GIORNO-GIOVANNA}
            Let $\D$ be a non-scalar diagonalizable matrix with strictly positive real eigenvalues. Then, there exists a stable matrix $\jac$ such that \cref{generalsystem} exhibits a Turing instability. Furthermore, if $n \geq 3$, then there exists a stable matrix $\jac$ such that \cref{generalsystem} exhibits a wave instability.
        \end{theorem}
        In all of these cases, we give a constructive proof for finding such matrices to induce specific instabilities and provide insights into how to tune parameters to design them to exhibit specific unstable wavelengths. We remark that having matrices inducing these instabilities then allows one to find bifurcation curves where marginal stability is crossed (i.e.~where $\max_\mu(\Re(\lambda)) = 0$ and a transversality condition is satisfied). One can then make use of amplitude equation formalisms to study properties of the solutions near these bifurcation curves \cite{TOMS,Villar-sepulveda-wave}, and we will explore this in the first example at the end of the paper.


        The key idea of the proof for \cref{th:main} follows the methods introduced in \cite{villar-sepulveda}. That is, by constructing a diagonalizable diffusion tensor, the eigenvalue problem in \cref{general_solvability_condition} is reduced to the analysis of a Jacobian matrix that looks like the ones treated in \cite{villar-sepulveda}, with the eigenvectors of the diffusion matrix allowing us to effectively control the linearization of the system. We then construct an unstable principal submatrix and set some eigenvalues of the diffusion matrix to zero to `pick out' this submatrix of the (non-scalar) Jacobian matrix, which is guaranteed to destabilize the full system in the limit as $\mu \to \infty$. This effectively allows us to get unstable growth rates $\lambda$ with $\Re(\lambda) > 0$ as $\mu \to \infty$, while all other growth rates will tend to $- \infty$ in this limit. In the end, this gives a way of constructing a suitable matrix $\D$ so that the linear system exhibits the desired instability.
        
        Next, one introduces a small positive parameter $\delta$ that `relaxes' these zero diagonal elements to be small positive numbers and, by continuity, arrives at a system with the correct dispersion relation; that is, there is a finite interval of $\mu$ for which $\Re(\lambda) > 0$, ensuring the desired instability. Turing or wave instabilities can be selected based on the imaginary part of these eigenvalues (chosen via the construction of the unstable principal submatrix above), and the range of $\mu$ exhibiting an instability can be tuned arbitrarily (see \cref{lemma:positiveeig} and \cref{app:useful_results}). The proof of \cref{th:GIORNO-GIOVANNA} is qualitatively similar but proceeds more directly. In particular, we design the Jacobian matrix to ensure the existence of a $2$ or $3$-component subsystem corresponding to at least two distinct diffusion rates to generate Turing or wave instabilities.

        \cref{fig:explanation} illustrates the key idea of the constructions for the case of a Turing instability. We construct a principal unstable submatrix that has the largest positive eigenvalue indicated by the green line. When $\delta = 0$, a dispersion curve of the full system (shown in orange) tends asymptotically to this value as $\mu \to \infty$. We then increase $\delta$ to make the diffusion matrix nonsingular. This leads to the blue curves, which have the appropriate form of dispersion relation to ensure a Turing instability with a finite maximally unstable mode. The red curve illustrates the critical value of $\delta > 0$, which gives this marginal stability criterion such that the dispersion curve is tangent to the horizontal axis. This represents the $\delta$-value at which a Turing bifurcation occurs. The blue dispersion curves thus represent a homotopy between the red and orange curves, parametrized by $\delta$. The proofs below will construct this homotopy explicitly, which will be illustrated in the examples.

        \begin{figure}
            \centering
            \includegraphics[width = \linewidth]{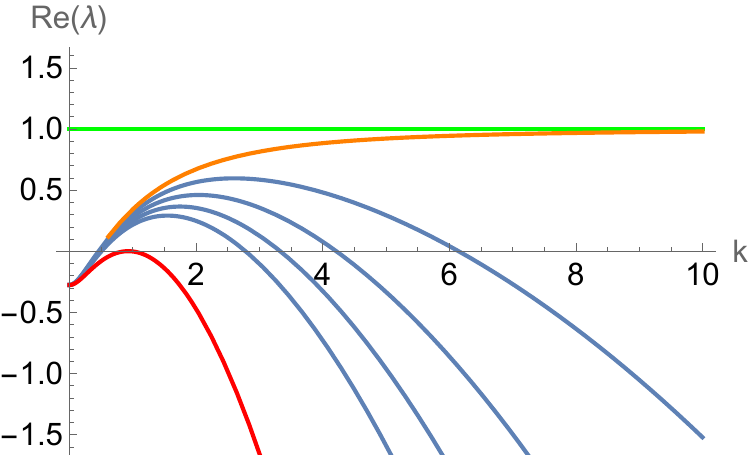}
            \caption{An example set of dispersion curves (i.e.~plots of $\Re\left(\lambda\left(k^2\right)\right) = \Re(\lambda(\mu))$) with different curves parameterized by a small constant diffusion parameter $\delta$. When $\delta = 0$ there exists a branch of the dispersion relation that follows the orange curve which has a prescribed growth rate in the limit as $k \to \infty$ that is given by the green line. For small discrete positive values of $\delta$, the blue curves are obtained and, as $\delta$ becomes larger, one eventually reaches the red curve, marking the boundary of the instability region, precisely where a bifurcation occurs.}
            \label{fig:explanation}
        \end{figure}
        
        Wave instabilities can also be designed similarly. Essentially, we have the same form of dispersion relations as illustrated in \cref{fig:explanation}, but with each curve now corresponding to the real part of a complex conjugate pair of $\lambda$ for each $k$. Transitions between Turing and wave instabilities as $k$ changes can also give rise to more intricate dispersion relations (see, for example, \cref{fig:disprelFHN}).

        Note that the constructions we use for the proofs of \cref{th:main} lead to matrices $\D$ that have strictly positive real-part eigenvalues. However, it should be noted that such matrices may not be symmetric, nor satisfy other functional thermodynamic constraints that are required for a particular application. Similarly, the Jacobian matrix constructed in the proof of \cref{th:GIORNO-GIOVANNA} will be of a particular form, which may not satisfy particular physical constraints on the dependent variables.  Nevertheless, the constructive proofs we develop are minimal in the sense that they provide exactly the core features of the matrices needed for the instabilities. Thus, in practice, one is left with many free parameters that can be chosen to satisfy additional modelling constraints.

    \subsection{Outline}
        The rest of the paper is outlined as follows: \cref{sec:proof1-general} is devoted to proving \cref{th:main} using ideas that largely extend \cite{villar-sepulveda}. \cref{sec:GIORNO-GIOVANNA} proves \cref{th:GIORNO-GIOVANNA} by considering the converse case of freedom in the Jacobian matrix, rather than the diffusion matrix exploited before. 
        We give examples of these theorems, and the particular constructions used to determine the instability-inducing matrices in \cref{sec:examples}, making use of a few additional results on matrix design collected in \cref{app:useful_results}. Finally, \cref{sec:conclusions} closes the article with some general remarks and possible extensions of the ideas presented here.
    
\section{Proof of \cref{th:main}} \label{sec:proof1-general}
    To prove \cref{th:main}, we will build some key lemmas extending those in \cite{villar-sepulveda} to the case of a cross-diffusion matrix (that is, for $\D$ being in general non-diagonal). A key step will be to construct a diagonalizable diffusion matrix $\D$ in order to use previous results applicable only for diagonal diffusion matrices, which we illustrate in the proof of \cref{lemma:positiveeig} below. Throughout, we will use subscripts $i_\ell$ to denote sequences of positive integer indices. We also define a matrix to be \emph{S-stable} if all of its principal submatrices are stable. Throughout the paper, the eigenvalues of $\D$ will be called $\sigma_i$, for $i = 1, \ldots, n$.

    \begin{lemma} \label{lemma:positiveeig}
        Let $\D$ be diagonalizable matrix; specifically, there exists a non-singular matrix $\m{Q}$ such that $\Db := \m{Q}^{- 1} \D \m{Q}$ is a diagonal matrix. Suppose further that $\D$ has $\ell$ eigenvalues $\sigma_{i_1}, \ldots, \sigma_{i_\ell}$ equal to zero, while all others are positive. Then, there exist $\ell$ eigenvalues of $\linop$ that tend to the eigenvalues of the principal submatrix of $\m{Q}^{- 1} \jac \m{Q}$ that comprises its $i_1^{\rm th}, \ldots, i_\ell^{\rm th}$ rows and columns as $\mu \to \infty$. Meanwhile, the real parts of the eigenvalues associated with the principal submatrix of $\m{Q}^{- 1} \jac \m{Q}$ that excludes these rows and columns all tend to $- \infty$ as $\mu \to \infty$.
    \end{lemma}
    \begin{proof}
        First, if $\ell = n$, then all the eigenvalues of $\D$ are equal to zero. Therefore, as $\D$ is diagonalizable, this implies it must be equal to the zero matrix, entailing that the eigenvalues of $\linop$ are equal to the eigenvalues of $\jac$ for every $\mu > 0$, which proves the result in this case.
    
        On the other hand, if $1 \leq \ell < n$, we have that the eigenvalues $\lambda_m$ of $\linop$ are determined by:
        \begin{align*}
            \jac \mbf v_m - \mu \D \mbf v_m - \lambda_m \mbf v_m = 0,
        \end{align*}
        where $\mbf v_m$ is a corresponding eigenvector for each $m = 1, \ldots, n$. Now, using that $\D = \m{Q} \Db \m{Q}^{- 1}$, we know that this equation is equivalent to
        \begin{align*}
            \jac \mbf v_m - \mu \m{Q} \Db  \m{Q}^{- 1} \mbf v_m = \lambda_m \mbf v_m,
        \end{align*}
        which implies
        \begin{align}
            \m{Q}^{- 1} \jac \mbf v_m - \mu \Db \m{Q}^{- 1} \mbf v_m = \lambda_m \m{Q}^{- 1} \mbf v_m. \label{eigvalprob}
        \end{align}
        Let $\mbf w_m = \m{Q}^{- 1} \mbf v_m$. Therefore, \cref{eigvalprob} becomes
        \begin{align}
            \m{Q}^{- 1} \jac  \m{Q} \mbf w_m - \mu \Db \mbf w_m = \lambda_m \mbf w_m, \label{eigenvalue_problem}
        \end{align}
        which is equivalent to the linearization of the general system analyzed in \cite{villar-sepulveda}, for $\m{Q}^{- 1} \jac \m{Q}$ instead of $\jac$. This implies that \cite[Lemma 5]{villar-sepulveda} can be used to conclude the proof.
    \end{proof}
    To use the lemma above, we are going to focus on proving the existence of stable matrices $\m{Q}^{- 1} \jac \m{Q}$ with unstable principal submatrices to conclude that the system can exhibit different diffusion-driven instabilities. \cref{lemma:positiveeig} tells us that, to construct the diffusion matrices that allow the appearance of these instabilities, we need to follow two steps: first, we find the matrix $\m{Q}$ of eigenvectors of $\D$ such that $\m{Q}^{- 1} \jac \m{Q}$ has an unstable principal submatrix; second, we change the eigenvalues of $\D$ in order to control the convergence of the eigenvalues of $\linop$ as $\mu \to \infty$ (i.e.~set the corresponding eigenvalues to asymptotically small values to induce the instability). While the case of zero diffusion coefficients corresponding to this principal unstable submatrix will lead to an ill-posed linear system, we can fix this by simply using small positive values instead of zeroes; see \cref{fig:explanation} for an illustration. An important point here is to understand what kind of structure we want to have for the unstable principle submatrix, and then use this structure to design the eigenvectors $\m{Q}$ of the diffusion matrix to ensure that this submatrix appears within $\m{Q}^{- 1} \jac \m{Q}$.
    
    To continue, it is useful to recall a standard result from linear algebra: if two matrices $\m{A}$ and $\m{B}$ have the same Jordan normal form, then $\m{A}$ and $\m{B}$ are similar. That is, there exists a non-singular matrix $\m{R}$ such that $\m{A} = \m{R} \m{B} \m{R}^{- 1}$ \cite{fraleigh}. Furthermore, we recall that if $\m{A}$ is a diagonalizable matrix of size $n \times n$, then $\m{A}$ has only one eigenvalue (with algebraic multiplicity $n$) if and only if $\m{A}$ is a scalar matrix (i.e.~a scalar multiple of the identity). We make use of the following procedure to reduce similarity in general to the question of similarity of principal submatrices arising from Jordan normal forms. We will also use the notation $\m{I_\ell}$ to stand for the $\ell$ by $\ell$ identity matrix, and $\m{0_{i, j}}$ for an $i \times j$ block of zeros. In both cases, these matrices are omitted when any of their sizes is lower than one (i.e.~if $\m{I_\ell}$ appears in an expression, and $\ell \leq 0$, then the matrix is omitted).

    \begin{remark} \label{remark_JNF}
        Let $\jac$ be similar to a matrix with $\ell \geq 1$ square blocks given by
        \begin{align*}
            \m{C} = \mbox{diag}\left(\m{C_1}, \ldots , \m{C_\ell}\right).
        \end{align*}
        Let us assume that there exists $1 \leq j \leq \ell$ such that $\m{C_j}$ is similar to $\m{K}$, with $\m{K}$ being a square matrix of the same size as $\m{C_j}$. This implies that there exists an invertible matrix $\m{R}$ such that $\m{C_j} \m{R} = \m{R} \m{K}$. Therefore, if we define $\m{Q} = \mbox{diag}\left(\m{I_1}, \ldots, \m{I_{j - 1}}, \m{R}, \m{I_{j + 1}}, \ldots, \m{I_\ell}\right)$, then we note that
        \begin{align*}
            \m{C} \m{Q} = \m{Q} \m{M},
        \end{align*}
        where
        \begin{align*}
            \m{M} = \mbox{diag}\left(\m{C_1}, \ldots, \m{C_{j - 1}}, \m{K}, \m{C_{j + 1}}, \ldots, \m{C_\ell}\right),
        \end{align*}
        implying that we can construct similar matrices only by establishing similarities of square blocks in the diagonal of the Jordan matrix of $\jac$.
    \end{remark}
    
    
    We now give a constructive proof for \cref{th:main}. The idea of the proof is to find a matrix $\D$ that ensures that the system exhibits either Turing or wave instabilities. In \cref{lemma:positiveeig}, we show that the existence of an unstable principal submatrix is a sufficient criterion to ensure the appearance of any kind of linear instability. The proof will proceed then, for each kind of instability, by constructing a matrix $\m{K}$ similar to $\jac$ with an unstable principal submatrix. In each case, we proceed by directly writing down a matrix $\m{K}$ with the right unstable submatrix for a particular kind of instability and then showing that this can be made similar to $\jac$ through the correct choice of $\D$. Note throughout that there will be several free parameters at the end of the construction, leaving freedom one could exploit to match any further constraints arising from the modelling context.
    
    \subsection{Turing instabilities} \label{sub:Turing-1}
        We start by providing the following result for 2-component systems.
        \begin{lemma} \label{lemma:n=2}
            Let $n = 2$, and
            $\jac$ be a non-scalar stable matrix. Then, there exists a diffusion matrix $\D$ so that \cref{generalsystem} exhibits a Turing instability.
        \end{lemma}
        \begin{proof}
            Let us denote the elements of $\jac$ by
            \begin{align*}
                \jac = \begin{pmatrix}
                    a & b
                    \\
                    c & d
                \end{pmatrix}.
            \end{align*}
            Consider the following matrix:
            \begin{align*}
                \m{K} &= \begin{pmatrix}
                    a + d - 1 && a \, d - b \, c - a - d + 1
                    \\
                    - 1 && 1
                \end{pmatrix},
            \end{align*}
            which clearly fulfills $\tr(\m{K}) = \tr\left(\jac\right)$ and $\det(\m{K}) = \det\left(\jac\right)$, implying that both matrices have the same eigenvalues.
        
            We now prove that these matrices are similar. In fact, consider
            \begin{align*}
                \m{Q} = \begin{pmatrix}
                    \alpha && (d - 1) \, \alpha - b \, \gamma
                    \\
                    \gamma && (a - 1) \, \gamma - c \, \alpha
                \end{pmatrix},
            \end{align*}
            for $\alpha, \gamma\in \mathbb R$. It is easy to see that
            \begin{align*}
                \jac \m{Q} &= \m{Q} \m{K},
            \end{align*}
            and
            \begin{align*}
                \det(\m{Q}) = b \, \gamma^2 + (a - d) \, \alpha \, \gamma - c \, \alpha^2,
            \end{align*}
            which is either a quadratic or a linear function on $\alpha$ or $\gamma$, as $a\neq d$, $b \neq 0$ or $c\neq 0$. This is true since, as $\jac$ is not a scalar matrix, then the elements in the diagonal must not be equal or the off-diagonal terms must not be zero at the same time. We can then take any pair $(\alpha, \gamma)$ so that $\det(\m{Q}) \neq 0$.
        
            This proves that $\jac$ is similar to $\m{K}$, which is a stable matrix with an unstable principal submatrix. Therefore, by choosing a suitable diagonal matrix $\Db$, we can construct the matrix $\D = \m{Q} \Db \m{Q}^{- 1}$ so that the result holds by using \cref{lemma:positiveeig} (i.e.~by choosing a sufficiently small bottom-right diagonal element to ensure a positive $\lambda$ for some $\mu > 0$, thanks to the positive bottom-right element of $\m{K}$).
        \end{proof}
        Using \cref{remark_JNF}, and the fact that any non-scalar matrix $\jac$ will have a $2\times 2$ non-scalar block in its Jordan Normal Form, we can immediately use \cref{lemma:n=2} to construct a suitable matrix $\D$ to induce Turing instabilities. The similarity matrix $\m{Q}$ constructed in the proof of the Lemma is the key to proving the general result. Note that $\m{K}$ was chosen in order to fulfill the conditions provided in \cref{lemma:positiveeig}, but there are infinitely many other matrices that are similar to $\jac$ and have an unstable principal submatrix.

\subsection{Wave instabilities} \label{sub:proof-wave}
    We will now use similar constructions to demonstrate how to achieve wave instabilities with suitable choices of $\D$.
    \begin{lemma} \label{lemma:wave}
        Let $\jac$ be a non-scalar stable matrix and $n = 3$ or $n = 4$. If $\jac$ has $1 \leq \ell \leq 3$ distinct eigenvalues, then there exists a diffusion matrix $\D$ so that \cref{generalsystem} exhibits a wave instability provided $\jac$ has at least $3 - \ell$ generalized eigenvectors.
    \end{lemma}
    \begin{proof}        
        We first consider the case $n = 3$. Let us denote the elements of $\jac$ by
        \begin{align*}
            \jac = \begin{pmatrix}
                a & b & c
                \\
                d & e & f
                \\
                g & h & i
            \end{pmatrix}.
        \end{align*}
        Consider the following matrix
        \begin{align*}
            \m{K} = \begin{pmatrix}
                \ell_{1,1} & \ell_{1, 2} & \ell_{1, 3}
                \\
                \alpha & 1 & - 1
                \\
                \gamma & 1 & 1
            \end{pmatrix},
        \end{align*}
        where $\alpha, \gamma \in \mathbb R$, and
        \begin{align*}
            \ell_{1, 1} &= a + e + i - 2,
            \\
            \ell_{1, 2} &= \frac{2 \alpha (a + e + i - 1) - 2 \gamma + \left(M_{12} + M_{13} + M_{23}\right) (\gamma -\alpha) - \gamma \det(\jac)}{\alpha^2 + \gamma^2},
            \\
            \ell_{1, 3} &= \frac{2 (\gamma (a + e + i - 1) + \alpha) - \left(M_{12} + M_{13} + M_{23}\right) (\alpha + \gamma) + \alpha \det(\jac)}{\alpha^2 + \gamma^2},
        \end{align*}
        with
        \begin{align*}
            M_{12} = a e - b d, \qquad M_{13} = a i - c g, \qquad M_{23} = e i - f h.
        \end{align*}
        Note that $\ell_{1, 2}$ and $\ell_{1, 3}$ are well-defined provided that $\alpha^2 + \gamma^2 \neq 0$.
        
        Now, we consider the case when $\jac$ has three distinct eigenvalues. By construction, we have that the characteristic polynomials of $\jac$ and $\m{K}$ are the same, which implies that these matrices are similar (enforcing this similarity is precisely how the $\ell_{1, i}$ were constructed). We highlight that the $2 \times 2$ submatrix at the bottom right corner of $\m{K}$ has eigenvalues $1 \pm i$. Therefore, by arguments similar to those used at the end of the proof of \cref{lemma:n=2}, we can construct a matrix $\D$ to conclude the first part of the proof using the result in \cref{lemma:positiveeig}.

        On the other hand, if $\jac$ has repeated eigenvalues, then all of them must be real (as any complex eigenvalue of a real matrix comes in pairs). This implies that $\jac$ is similar to one of the matrices below:
        \begin{align*}
            \m{C_1} = \begin{pmatrix}
                - y & 0 & 0
                \\
                0 & - z & \ell
                \\
                0 & 0 & - z
            \end{pmatrix}, \quad \m{C_2} = \begin{pmatrix}
                - z & \ell_2 & 0
                \\
                0 & - z & \ell_1
                \\
                0 & 0 & - z
            \end{pmatrix}
        \end{align*}
        where $y, z > 0$ are real numbers so that $y \neq z$, and $\ell, \ell_1, \ell_2 \in \{0, 1\}$. In this case, we can use the same form of the matrix $\m{K}$ given above but need to find a similarity matrix $\m{R}$ to show that $\jac$ is similar to $\m{C_1}$ or $\m{C_2}$.

        Firstly, we set 
        \begin{align*}
            \ell_{1, 1} &= - y - 2 z - 2,
            \\
            \ell_{1, 2} &= - \frac{2 \alpha + 2 \gamma + 2 \alpha y - \gamma y z^2 + 2 \alpha y z - 2 \gamma y z + \alpha z^2 - \gamma z^2 + 4 \alpha z}{\alpha^2 + \gamma^2},
            \\
            \ell_{1, 3} &= - \frac{- 2 \alpha + 2 \gamma + 2 \gamma y + \alpha y z^2 + 2 \alpha y z + 2 \gamma y z + \alpha z^2 + \gamma z^2 + 4 \gamma z}{\alpha^2 + \gamma^2},
        \end{align*}
        which implies that $\jac$ and $\m{C_1}$ have the same eigenvalues. With this, we note that these matrices are similar if there exists an invertible matrix $\m{R} = \left(r_{i, j}\right)_{i, j}$ so that $\m{C_1} \m{R} = \m{R} \m{K}$. We omit the solution of this system here, but we highlight that, if we leave $r_{1, 1}, r_{2, 1}$ and $r_{3, 1}$ as parameters of the solution (solving for the other 6 elements of $\m{R}$), then the determinant of $\m{R}$ will be given by
        \begin{align*}
            - \frac{\ell r_{1, 1} r_{3, 1}^2 (y - z)^2}{\alpha^2 + \gamma^2}.
        \end{align*}
        See \cref{app:R_similarity}, where we give the full form of $\m{R}$, ensuring that it can always be found. With this, we note that $\m{R}$ will be invertible if and only if $\ell \neq 0$ and $r_{1, 1}, r_{3, 1} \neq 0$. This implies that $\ell = 1$ and, as the similarity of matrices is transitive, we have again constructed a matrix $\m{K}$ similar to $\jac$ with an unstable principal submatrix. Hence, we can conclude that a wave instability can be generated when $\jac$ has only two distinct eigenvalues with one generalized eigenvector.
        
        In the case of similarity to $\m{C_2}$, we need to set
        \begin{align*}
            \ell_{1, 1} &= - 3 z - 2,
            \\
            \ell_{1, 2} &= \frac{\gamma \left(z^3 + 3 z^2 - 2\right) - \alpha \left(3 z^2 + 6 z + 2\right)}{\alpha^2 + \gamma^2},
            \\
            \ell_{1, 3} &= - \frac{\gamma \left(3 z^2 + 6 z + 2\right) + \alpha \left(z^3 + 3 z^2 - 2\right)}{\alpha^2 + \gamma^2}\ak{.}
        \end{align*}
        As before, we note that $\m{K}$ will be similar to $\m{C_2}$ if there exists a matrix $\m{R} = \left(r_{i, j}\right)_{i, j}$ such that $\m{C_2} \m{R} = \m{R} \m{K}$. In this case, if we again leave $r_{1, 1}, r_{2, 1}$ and $r_{3, 1}$ as parameters of the solution (solving for the other 6 elements of $\m{R}$), then the determinant of $\m{R}$ will be given by
        \begin{align*}
            - \frac{\ell_1^2 \ell_2 r_{3, 1}^3}{\alpha^2 + \gamma^2},
        \end{align*}
        which is different from zero as long as $\ell_1 \ell_2 r_{3, 1} \neq 0$ and $\alpha^2 + \gamma^2 \neq 0$, which implies that $\ell_1 = \ell_2 = 1$. Again, see \cref{app:R_similarity} where we show that such an $\m{R}$ can always be found. Hence, by the same argument given above, we can conclude that the instability can be created when $\jac$ has only one eigenvalue (with multiplicity three) with two generalized eigenvectors.

        In the case $n = 4$, we can use \cref{remark_JNF} to reduce this to a $3 \times 3$ block of the Jordan normal form of $\jac$, unless $\jac$ has repeated complex conjugate eigenvalues (note that a repeated real root and a complex-conjugate pair can be subsumed in the case of 3 distinct eigenvalues). The only case not covered by the arguments above then, is if there is a repeated complex-conjugate pair of eigenvalues. In this case, we consider
        \begin{align*}
            \m{K} = \begin{pmatrix}
                \ell_{1, 1} & \ell_{1, 2} & \ell_{1, 3} & \ell_{1, 4}
                \\
                \gamma & \alpha & 0 & 0
                \\
                \delta & 0 & - \alpha & \alpha
                \\
                \varepsilon & 0 & - \alpha & - \alpha
            \end{pmatrix}, \qquad \text{and} \quad \qquad \m{C_3} = \begin{pmatrix}
                \alpha & - \beta & \ell & 0
                \\
                \beta & \alpha & 0 & \ell
                \\
                0 & 0 & \alpha & - \beta
                \\
                0 & 0 & \beta & \alpha
            \end{pmatrix},
        \end{align*}
        where $\alpha < 0$, and $\beta > 0$ are real numbers and $\ell \in \{0, 1\}$. In this case, we note (via direct calculation) that the characteristic polynomials of $\m{K}$ and $\m{C_3}$ are equal if
        \begin{align*}
            \ell_{1, 1} &= 5 \alpha,
            \\
            \ell_{1, 2} &= - \frac{\beta^4}{5 \alpha^2 \gamma },
            \\
            \ell_{1,3} &= \frac{\alpha^4 (10 \varepsilon - 55 \delta ) - 10 \alpha^2 \beta^2 (\delta - 2 \varepsilon) + \beta^4 (\delta + 2 \varepsilon)}{5 \alpha ^2 \left(\delta^2 + \varepsilon^2\right)},
            \\
            \ell_{1, 4} &= \frac{- 5 \alpha^4 (2 \delta + 11 \varepsilon) - 10 \alpha^2 \beta^2 (2 \delta + \varepsilon) + \beta^4 (\varepsilon - 2 \delta)}{5 \alpha^2 \left(\delta^2 + \varepsilon^2\right)}.
        \end{align*}
        This implies that we have the constraints $\gamma \neq 0$ and $\delta^2 + \varepsilon^2 \neq 0$ in order to be able to make these matrices similar. So, as before, we note that $\m{K}$ will be similar to $\m{C_3}$ if there exists a matrix $\m{R} = \left(r_{i, j}\right)_{i, j}$ such that $\m{C_3} \m{R} = \m{R} \m{K}$. In this case, if we leave $r_{1, 1}, r_{2, 1}$, $r_{3, 1}$ and $r_{4, 1}$ as parameters of the solution (and solve for the other $12$ elements of $\m{R}$), then the determinant of $\m{R}$ will be given by
        \begin{align*}
            - \frac{4 \beta^4 \ell^2 \left(r_{3, 1}^2 + r_{4, 1}^2\right)^2}{5 \alpha^3 \gamma \left(\delta^2 + \varepsilon ^2\right)},
        \end{align*}
        which is different from zero provided that $\ell = 1$, and $r_{3, 1}^2 + r_{4, 1}^2 \neq 0$. Note that the $2 \times 2$ bottom-right submatrix of $\m{K}$ has eigenvalues $- \alpha \pm \alpha i$, which have a positive real part, as $\alpha < 0$. As before, we can therefore construct a matrix $\D$ to conclude the proof using the result in \cref{lemma:positiveeig}.

        Once again, we can conclude that a wave instability can be generated when $n = 4$ and $\jac$ has only two distinct complex-conjugate eigenvalues as long as it has two generalized eigenvectors.
        
    \end{proof}
    We remark that the case $n = 3$ is minimal except that it does not account for repeated complex-conjugate eigenvalues, hence why we included the case $n = 4$ in the statement of the previous lemma. We also note that other matrices $\m{K}$ can be chosen to arrive at similar results via different constructions. For instance, in the $n = 4$ case above, the entries or the location of the principal submatrix at the bottom right corner of $\m{K}$ does not change the result, and the elements of the second row and column of the $4 \times 4$ matrix $\m{K}$ can also be non-zero and the result will be the same (with a suitable modification of $\m{R}$ and $\D$).
    
    With this, we are ready to prove the first main result of this article.
    \begin{proof}[Proof of \cref{th:main}]
        First, we consider how to induce Turing instabilities. Given that $\jac$ is non-scalar, we know there must exist a non-scalar $2 \times 2$ sub-block of its Jordan normal form which is also non-scalar. We can apply \cref{lemma:n=2} to find a matrix $\m{K}$ similar to this sub-block, and by construction, this $\m{K}$ will have a principle unstable submatrix that induces a Turing instability. We then use the construction given in \cref{remark_JNF} to embed this similarity inside of the matrix $\m{M}$ in that remark, and then apply \cref{lemma:positiveeig} on the matrix $\m{Q}^{- 1} \jac \m{Q}$, which now has a principal unstable submatrix by construction. This Lemma allows us to set one of the eigenvalues of the diffusion matrix (i.e.~an element of $\Db$) to zero in order to excite the corresponding unstable principal submatrix and induce a Turing instability in the full system. Note that we can take all other elements of $\Db$ as any arbitrary positive numbers in order to use \cref{lemma:positiveeig}.
    
        The proof for inducing wave instabilities proceeds in an analogous way as above, except that we will make use of \cref{lemma:wave} to construct a matrix $\m{K}$ similar to $\jac$, instead. By assumption, $\jac$ has $1 \leq \ell \leq 3$ distinct eigenvalues with at least $3 - \ell$ corresponding generalised eigenvectors. Hence, either $\jac$ has a $3 \times 3$ sub-block of its Jordan normal form that exactly fulfills this criterion as well, or it has a $4 \times 4$ sub-block with two repeated complex conjugate eigenvalues. In either of these cases, we can again embed the appropriate matrix $\m{K}$ from the proof of \cref{lemma:wave} into the matrix $\m{M}$ from above in order to construct a suitable similarity matrix $\m{Q}$ such that $\m{Q}^{- 1} \jac \m{Q}$ has a principal unstable submatrix with complex-conjugate unstable eigenvalues. Again, through an analogous use of \cref{lemma:positiveeig}, we can set the appropriate two elements of $\Db$ to zero to induce a wave instability in the full system. This concludes the proof.

    \end{proof}
    We remark that, in practice, one needs the zero elements of $\Db$ to be asymptotically small rather than strictly zero to avoid ill-posedness of the linear problem, but the existence of the instability will persist for some range of `small' diffusion constants. We note that in both the diagonal matrices $\Db$ and in the similarity matrices $\m{Q}$ constructed above, one has, in principle, freedom over elements not affecting the principal unstable submatrix to use more general constructions to gain additional flexibility over the elements we have not exploited here. We will illustrate some aspects of this in \cref{app:useful_results} through alternative lemmas one can use to construct these matrices, and in the examples in \cref{sec:examples}.

\section{Proof of \cref{th:GIORNO-GIOVANNA}} \label{sec:GIORNO-GIOVANNA}
    As in the previous section, we prove \cref{th:GIORNO-GIOVANNA} by dividing it into cases, though as the eigenvalues of $\D$ are all real, the number of sub-cases we need to consider is fewer than in the proof of \cref{th:main}.

    We assume, without loss of generality, that the eigenvalues of $\D$ (i.e.~the elements of $\Db$) are sorted in ascending order so that $0 < \sigma_i \leq \sigma_{i + 1}$ for all $i = 1, \ldots, n - 1$.
    \subsection{Turing instabilities}
        \begin{lemma} \label{lemma:thm2turing}
            Let $n \geq 2$ and $\D$ be a non-scalar diagonalizable matrix with real and strictly positive eigenvalues. Then, there exists a stable matrix $\jac$ such that \cref{generalsystem} exhibits a Turing instability.
        \end{lemma}
        \begin{proof}
            As $\D$ is a diagonalizable non-scalar matrix with positive real eigenvalues and the elements of its diagonalization $\mathbb D$ are sorted, then there must exist at least one integer $1 \leq \ell \leq n - 1$ such that $\sigma_\ell < \sigma_{\ell + 1}$.

            Therefore, as the eigenvalues of $\linop$ are the same as the ones of $\m{Q}^{- 1} \jac \m{Q} - \mu \mathbb D$, we consider
            \begin{align*}
                \jac = \m{Q} \m{C} \m{Q}^{- 1},
            \end{align*}
            where
            \begin{align*}
                \m{C} = \begin{pmatrix}
                    - \m{I_{\ell - 1}} & \m{0_{\ell - 1, 2}} & \m{0_{\ell - 1, n - \ell - 1}}
                    \\
                    \m{0_{2, \ell - 1}} & \m{S^*} & \m{0_{2, n - \ell - 1}}
                    \\
                    \m{0_{n - \ell - 1, \ell - 1}} & \m{0_{n - \ell - 1, 2}} & - \m{I_{n - \ell - 1}}
                \end{pmatrix},
            \end{align*}
            with
            \begin{align*}
                \m{S^*} = \begin{pmatrix}
                    b & b (b + c) + q
                    \\
                    - 1 & - (b + c)
                \end{pmatrix},
            \end{align*}
            and $b, c, q > 0$ are real constants to be determined. By checking its trace and determinant, one can see that $\m{S^*}$ is stable. By construction, we have that $n - 2$ eigenvalues of $\m{Q}^{- 1} \jac \m{Q} - \mu \Db$ have the form $- 1 - \sigma_i$, for a corresponding integer $1 \leq i \leq n$ such that $i \neq \ell, \ell + 1$. On the other hand, by considering the two eigenvalues we have singled out via $\m{S^*}$, we see that the equations that determine them are uncoupled from all others and, hence, they form an $n = 2$ linear system which is precisely the classical Turing system, with $\m{S^*}$ playing the role of the $2 \times 2$ Jacobian matrix. Thus, the largest of the other two eigenvalues will be positive for some value of $\mu > 0$ if the system satisfies the classical 2-component Turing conditions (see, e.g., \cite{Murray2001}):
            \begin{align*}
                \sigma_{\ell + 1} b - \sigma_\ell (b + c) &> 0, \quad \text{and}
                \\
                \left(\sigma_{\ell + 1} b - \sigma_\ell (b + c)\right)^2 &> 4 \sigma_\ell \sigma_{\ell + 1} q,
            \end{align*}
            which hold provided that
            \begin{align*}
                b > \frac{\sigma_\ell c}{\sigma_{\ell + 1} - \sigma_\ell} > 0
            \end{align*}
            and
            \begin{align*}
                q < \frac{\left(\sigma_{\ell + 1} b - \sigma_\ell (b + c)\right)^2}{4 \sigma_\ell \sigma_{\ell + 1}}.
            \end{align*}
            Hence, we can always pick values of $b, c,$ and $q$ to satisfy these, given any pair of distinct eigenvalues $\sigma_\ell < \sigma_{\ell + 1}$. Therefore, if we set $\jac$ using this construction, we will induce a Turing instability in the original system given by \cref{generalsystem}.
        \end{proof}

    \subsection{Wave instabilities}

        We next prove an analogous result to the previous lemma for the case of wave instabilities. Due to these requiring a 3-component subsystem, the constructions are a bit more involved.
        \begin{lemma}\label{lemma:thm2wave}
            Let $n \geq 3$ and $\D$ be a non-scalar diagonalizable matrix with real and strictly positive eigenvalues. Then, there exists a stable matrix $\jac$ such that the full system exhibits a wave instability.
        \end{lemma}


        \begin{proof}
            As $\D$ is a non-scalar diagonalizable matrix with positive real eigenvalues and the elements in $\mathbb D$ are sorted, then there must exist at least one integer $1 \leq \ell \leq n - 2$ such that $\sigma_\ell < \sigma_{\ell + 2}$.

            Let us consider $\jac = \m{Q} \m{C} \m{Q}^{- 1}$, where
            \begin{align*}
                \m{C} = \begin{pmatrix}
                    - \m{I_{\ell - 1}} & \m{0_{\ell - 1, 3}} & \m{0_{\ell - 1, n - \ell - 2}}
                    \\
                    \m{0_{3, \ell - 1}} & \m{S^*} & \m{0_{3, n - \ell - 2}}
                    \\
                    \m{0_{n - \ell - 2, \ell - 1}} & \m{0_{n - \ell - 2, 3}} & - \m{I_{n - \ell - 2}}
                \end{pmatrix},
            \end{align*}
            with
            \begin{align*}
                \m{S^*} = \begin{pmatrix}
                    b & - b & c y - b (2 b + c) 
                    \\
                    b & b & b c + y (2 b + c) + z
                    \\
                    1 & - 1 & - (2 b + c)
                \end{pmatrix},
            \end{align*}
            and $b, c, y, z > 0$ are real parameters to be determined. Note that $\tr{\m{S^*}} = - c < 0$. Furthermore, the eigenvalues of the $2 \times 2$ submatrix at the top left corner of $\m{S^*}$ are given by $b \pm b i$. We also note that $\m{S^*}$ is a stable matrix, as its characteristic polynomial for an eigenvalue $\eta$ is given by
            $$
                \eta^3 + c \eta^2 + \left(2 b y + z\right) \eta + 2 b c y = 0,
            $$
            which satisfies the Routh-Hurwitz stability conditions.
            
            Now, when we consider the full system, we note that (by the construction of $\m{C}$) $n - 3$ eigenvalues of $\m{Q}^{- 1} \jac \m{Q} - \mu \Db$ have the form $- 1 - \sigma_i$, for a corresponding integer $1 \leq i \leq n$ such that $i \neq \ell, \ell + 1, \ell + 2$. Moreover, the other three eigenvalues are solutions to the following equation:
            \begin{align*}
                - \det\left(\m{S^*} - \mu \mbox{diag}\left(\sigma_\ell, \sigma_{\ell + 1}, \sigma_{\ell + 2}\right) - \lambda \m{I}\right) = 0,
            \end{align*}
            which is given by
            \begin{equation} \label{eq_wave_lambda}
                \lambda^3 + p_2 \lambda^2 + p_1 \lambda + p_0 = 0,
            \end{equation}
            where
            \begin{equation*}
                \begin{aligned}
                    p_2 &= c + \mu \left(\sigma_\ell + \sigma_{\ell + 1} + \sigma_{\ell + 2}\right),
                    \\
                    p_1 &= \mu \sigma_{\ell + 1} \left(b + c + \mu \sigma_{\ell + 2}\right) + \mu \sigma_\ell \left(b + c + \mu \left(\sigma_{\ell + 1} + \sigma_{\ell + 2}\right)\right) - 2 b \mu \sigma_{\ell + 2} + 2 b y + z,
                    \\
                    p_0 &= \mu \sigma_\ell \left(- 2 b^2 + \mu \sigma_{\ell + 1} \left(2 b + c + \mu \sigma_{\ell + 2}\right) - b \mu \sigma_{\ell + 2} + 2 b y + c y + z\right)
                    \\
                    & \quad + \left(2 b - \mu \sigma_{\ell + 1}\right) \left(b \mu \sigma_{\ell + 2} + c y\right).
                \end{aligned}
            \end{equation*}
            We will now show that we can ensure the existence of an imaginary root with a positive real part $\tilde \lambda = \beta + \omega i$ for some values of $\beta, \mu, \omega > 0$ through particular choices of the parameters in $\m{S^*}$. In particular, substituting $\lambda$ into \cref{eq_wave_lambda}, we find that $\omega$ needs to satisfy
            \begin{align}
                r_0 + r_2 \omega^2 + i \omega \left(r_i - \omega^2\right) = 0, \label{messy-eq}
            \end{align}
            where
            \begin{align*}
                r_0 &= 2 \mu b^2 \left(\sigma_{\ell + 2} - \sigma_\ell\right) + b \left(2 y (c + \beta) + \mu \sigma_\ell \left(2 \mu \sigma_{\ell + 1} - \mu \sigma_{\ell + 2} + 2 y + \beta\right) - 2 \mu \beta \sigma_{\ell + 2}\right.
                \\
                & \quad \left. + \mu \sigma_{\ell + 1} \left(\beta - \mu \sigma_{\ell + 2}\right)\right) + \mu \sigma_{\ell + 1} \left(- c y + c \beta + \mu \beta \sigma_{\ell + 2} + \beta^2\right)
                \\
                & \quad + \mu \sigma_\ell \left(\mu \sigma_{\ell + 1} \left(c + \mu \sigma_{\ell + 2} + \beta\right) + c (y + \beta) + \mu \beta \sigma_{\ell + 2} + z + \beta^2\right)
                \\
                & \quad  + \beta \left(\beta (c + \beta) + \mu \beta \sigma_{\ell + 2} + z\right),
                \\
                r_2 &= - c - \mu \left(\sigma_\ell + \sigma_{\ell + 1} + \sigma_{\ell + 2}\right) - 3 \beta,
                \\
                r_i &= b \left(2 y + \mu \left(\sigma_\ell + \sigma_{\ell + 1} - 2 \sigma_{\ell + 2}\right)\right) + \mu \sigma_{\ell + 1} \left(c + \mu \sigma_{\ell + 2} + 2 \beta\right)
                \\
                & \quad + \mu \sigma_\ell \left(c + \mu \left(\sigma_{\ell + 1} + \sigma_{\ell + 2}\right) + 2 \beta\right) + 2 c \beta + 2 \mu \beta \sigma_{\ell + 2} + z + 3 \beta^2.
            \end{align*}
            We will proceed to analyze solutions for large $b$ (i.e.~asymptotically in the limit of $b \gg 1$) in order to simplify the analysis. Throughout, we assume that all the other parameters are fixed and do not depend on $b$ in the asymptotics. To ensure that \cref{messy-eq} has a solution for $\omega^2$, we must have $r_i > 0$. We note that $r_i$ has a leading-order coefficient in $b$ given by
            \begin{align*}
                2 y + \mu \left(\sigma_\ell + \sigma_{\ell + 1} - 2 \sigma_{\ell + 2}\right),
            \end{align*}
            which is positive if and only if
            \begin{align}\label{y_bound}
                y > \frac{\mu}{2} \left(2 \sigma_{\ell + 2} - \sigma_\ell - \sigma_{\ell + 1}\right) > 0.
            \end{align}
            Furthermore, to ensure the real part of \cref{messy-eq} vanishes for $\omega^2 > 0$, we need that $r_0 r_2 < 0$. This is automatically true when $b$ is sufficiently large as the leading-order coefficient of $r_0 r_2$ is given by
            \begin{align*}
                - 2 \mu \left(\sigma_{\ell + 2} - \sigma_\ell\right) \left(c + \mu \left(\sigma_\ell + \sigma_{\ell + 1} + \sigma_{\ell + 2}\right) + 3 \beta\right) < 0,
            \end{align*}
            where we have used the ordering of the eigenvalues of $\Db$, i.e.~that $\sigma_\ell < \sigma_{\ell + 2}$.

            Let $q > 0$ be a real number such that
            \begin{align}
                y = \frac{\mu}{2} \left(2 \sigma_{\ell + 2} - \sigma_\ell - \sigma_{\ell + 1}\right) + q.
            \end{align}
            
            Setting the real and imaginary parts of \cref{messy-eq} equal to zero, we see that we need $\omega^2 = r_i = -r_0/r_2$. Writing this out, we arrive at the following equation for $\mu$:
            \begin{align}
                a_3 \mu^3 + a_2 \mu^2 + a_1 \mu + a_0 = 0, \label{mueq}
            \end{align}
            where
            \begin{align*}
                a_3 &= \left(\sigma_\ell + \sigma_{\ell + 1}\right) \left(\sigma_\ell + \sigma_{\ell + 2}\right) \left(\sigma_{\ell + 1} + \sigma_{\ell + 2}\right),
                \\
                a_2 &= b \left(\sigma_\ell - \sigma_{\ell + 1}\right) \left(\sigma_\ell - \sigma_{\ell + 2}\right) + \left(\frac{3 c}{2} + 2 \beta\right) \sigma_\ell^2 + 2 \beta \sigma_{\ell + 2}^2
                \\
                & \quad + \sigma_\ell \left(2 (c + 3 \beta) \sigma_{\ell + 1} + (c + 6 \beta) \sigma_{\ell + 2}\right) + \frac{1}{2} (c + 4 \beta) \sigma_{\ell + 1}^2 + 3 (c + 2 \beta) \sigma_{\ell + 1} \sigma_{\ell + 2},
                \\
                a_1 &= 2 b^2 \left(\sigma_\ell - \sigma_{\ell + 2}\right) + b \left((c + 2 q) \sigma_{\ell + 1} + 2 (q - c) \sigma_{\ell + 2} + c \sigma_\ell\right) + \sigma_{\ell + 2} (4 \beta  (c + 2 \beta) + z)
                \\
                & \quad + \sigma_\ell \left(c^2 - c q + 6 c \beta + 8 \beta^2\right) + \sigma_{\ell + 1} \left(c (c + q) + 6 c \beta + z + 8 \beta^2\right)
                \\
                a_0 &= 4 q \beta b + (c + 2 \beta) (2 \beta (c + 2 \beta) + z).
            \end{align*}
            Here, we note that the coefficient of $b^2$ in $a_1$ is given by $\sigma_\ell - \sigma_{\ell + 2}$, which is negative again due to the ordering of these eigenvalues. Hence, for sufficiently large $b$, the coefficient $a_1$ will become negative, implying (by the Routh-Hurwitz criterion) that there exists a root of \cref{mueq} with a positive real part. We need to ensure that, in the limit of large $b$, there is a positive real root, which is a necessary property for $\mu = k^2$ as an eigenvalue of the negative Laplacian.
            
            The discriminant of \cref{mueq} is given by
            \begin{align*}
                18 a_3 a_2 a_1 a_0 - 4 a_2^3 a_0 + a_2^2 a_1^2 - 4 a_3 a_1^3 - 27 a_3^2 a_0^2,
            \end{align*}
            and if it is positive, then we know that all of the roots of \cref{mueq} are real \cite[Section 10.3]{irving2004integers}. Hence, for sufficiently large $b$, we would have a positive real root. The full expression of the discriminant is unpleasant, so we again simplify things by considering $b \gg 1$ as an asymptotically large parameter. This leads to a lengthy but not overly tedious calculation as, for large $b$, it is relatively straightforward to see that one can work with the leading-order terms in $b$ for each of the coefficients and no terms involving $a_0$ appear in the leading-order expression. We then find a leading-order coefficient of the discriminant in $b$ (specifically, collecting terms of order $b^6)$ given by\footnote{We also used a computer algebra system to compute the coefficient of $b^6$ directly for the full expression of the discriminant to find that it matches, but we do not recommend trying to do this by hand.}
            \begin{align*}
                4 \left(\sigma_{\ell + 2} - \sigma_\ell\right)^3 h
            \end{align*}
            where
            \begin{align*}
                 h &= \left(10 \sigma_{\ell + 1} - \sigma_\ell\right) \sigma_\ell^2 + 7 \sigma_\ell \sigma_{\ell + 1} \left(\sigma_{\ell + 1} +  2 \sigma_{\ell + 2}\right) + 8 \left(\sigma_\ell + \sigma_{\ell + 1}\right) \sigma_{\ell + 2}^2
                 \\
                 & \quad + 9 \left(\sigma_\ell^2 + \sigma_{\ell + 1}^2\right) \sigma_{\ell + 2},
            \end{align*}
            which is positive as we have that $0 < \sigma_\ell \leq \sigma_{\ell + 1} \leq \sigma_{\ell + 2}$ and $\sigma_\ell < \sigma_{\ell + 2}$. This implies that the root with a positive real part we had found for \cref{mueq} is real for a sufficiently large value of $b$, concluding that there exists a positive real value of $\mu$ such that $\lambda(\mu) = \beta + i \omega$, where $\beta$ can be taken to be positive, thus ensuring the existence of a wave instability.
        \end{proof}

        \begin{remark}
            While the above lemma uses a specific form of a Jacobian matrix, in principle one can construct a stable matrix of a different form by putting another unstable submatrix at the top left corner of $\m{S^*}$ by then making sure one ensures the existence of zeroes of \cref{mueq} for some positive values of $\mu$. Even with the form of the Jacobian block used, one can freely vary $c$, $q$, and $z$, needing to only ensure that $b$ is sufficiently large. One can compute an explicit minimum value for $b$, but it is somewhat tedious and easier to check that $b$ is sufficiently large by directly computing the dispersion relation (i.e.~plotting $\Re(\lambda(\mu))$).
        \end{remark}

        \begin{proof}[Proof of \cref{th:GIORNO-GIOVANNA}]
            The proof is analogous to the proof of \cref{th:main}, though simpler as we can directly use \cref{lemma:thm2turing,lemma:thm2wave} to design unstable submatrices of $\jac$ fulfilling the criteria for each kind of instability. 
        \end{proof}
        We remark that, as long as one ensures that $\jac$ is a stable matrix, there is otherwise freedom to design all other nonzero elements of the matrix away from those that ensure these specific unstable submatrices. With more effort, one can modify other elements of $\jac$, but will need to use more general constructions than those proposed here.

\section{Examples} \label{sec:examples}
    We now employ the constructions used in the proofs of the main theorems, as well as some additional results collected in \cref{app:useful_results} to illustrate the design of Turing and wave instabilities in a handful of prototypical reaction-cross-diffusion systems. The examples here demonstrate the key ingredients needed to use the results and how the freedom in the constructions allows one to satisfy additional constraints and still achieve the correct instabilities.

    We denote $\m{S_{i_1 \ldots i_\ell}}$ the principle submatrix of $\jac$ composed of the rows and columns with labels $i_1, \ldots, i_\ell$, where $1 \leq i_1 <\ldots < i_\ell \leq n$ is a set of non-repeated integers. We will plot dispersion relations (i.e.~$\Re(\lambda(\mu))$) with continuous (respectively, dashed) lines representing purely real (respectively, complex non-real) eigenvalues. In other words, we will not plot the imaginary parts of $\lambda(\mu)$ but instead use a dashed line for a given value of $\mu$ to indicate a complex growth rate $\lambda$, showing only its real part.

    We compute direct numerical solutions of the PDE models implemented using finite differences in space with the 3-point (1D) and 5-point (2D) standard Laplace stencil (or a suitable centred-difference approach for the case of a nonlinear diffusion coefficient), and then integrate in time using the Matlab function \texttt{ode15s}, which implements backwards differentiation formulae of orders 1 to 5. We used relative and absolute tolerances of $10^{- 9}$ and also checked numerical convergence in decreasing spatial step sizes for some specific simulations. In 1D (respectively, 2D) simulations, we used $1,000$ (respectively, $200 \times 200$) discrete nodes. Throughout, we use initial data of the form $\m{u} \rvert_{t = 0} = \m{P} + \eta \mathcal{N}(0, 1)$, which are small normally-distributed random perturbations of the initial condition with variance $\eta$. In all simulations, we used $\eta = 0.01$. All of the codes used for time integration can be found on GitHub\footnote{\url{https://github.com/AndrewLKrause/Designing-reaction-cross-diffusion-systems}}, and there we also include four links to interactive web-based versions of the simulations using VisualPDE \cite{walker2023visualpde}.

    \subsection{3-component Schnakenberg system}
        The Schnakenberg system was intended to be a simple 2-component model for glycolysis exhibiting limit cycles \cite{Schnakenberg1979,FahadWoods,villar2023degenerate}, and is well-studied in pattern formation largely because of its simplistic form and presentation in the textbook \cite{Murray2001}. Here, in order to exhibit wave instabilities, we consider a simple 3-component extension obtained after adding a third variable to the classical model as in \cite{xie2021complex}. However, we also include linear cross-diffusion terms not accounted for there. The reaction-diffusion system is given by
        \begin{align}
            \partial_t \u = \mbf f(\u) + \D \partial_{xx} \u, \label{Schnakenberg}
        \end{align}
        where
        \begin{align*}
            \u &= (u, v, w)^\intercal, \quad \mbf f(\u) = \begin{pmatrix}
                a - u + u^2 v
                \\
                b - u^2 v
                \\
                u - w
            \end{pmatrix}, \quad \text{and} \quad  \D = \begin{pmatrix}
                D_{1, 1} & D_{1, 2} & D_{1, 3}
                \\
                D_{2, 1} & D_{2, 2} & D_{2, 3}
                \\
                D_{3, 1} & D_{3, 2} & D_{3, 3}
            \end{pmatrix}.
        \end{align*}
        This system has only one homogeneous steady state given by
        \begin{align*}
            \mbf P = \left(u^*, v^*, w^*\right) = \left(a + b, \frac{b}{(a + b)^2}, a + b\right).
        \end{align*}
        
        In this example, we will assume that we can vary the diffusion matrix to obtain wave instabilities, without any further constraints. Hence, we will be using ideas from the proof of \cref{th:main}, which are presented in \cref{sec:proof1-general}. To be concrete, consider the parameter values $(a, b) = (1, 1/2)$. The Jacobian matrix of the system at $\mbf P$ is then
        \begin{align*}
            \jac = \left(
                \begin{array}{ccc}
                     - \frac{1}{3} & \frac{9}{4} & 0
                     \\[1ex]
                     - \frac{2}{3} & - \frac{9}{4} & 0
                     \\[1ex]
                     1 & 0 & - 1
                \end{array}
                \right),
        \end{align*}
        which has eigenvalues $- 1$ and $\frac{1}{24} \left(- 31 \pm \sqrt{335} i\right)$. \cref{tab:subeigval} shows the eigenvalues of the principal submatrices of $\jac$.
        
        {
        \newcolumntype{L}{>{$}l<{$}} 
        \begin{table}
            \centering
            \begin{tabular}{|L|L|}
                \hline
                \text{Submatrix} & \text{Eigenvalues}
                \\
                \hline
                \m{S_1} & - 1/3
                \\
                \hline
                \m{S_2} & - 9/4
                \\
                \hline
                \m{S_3} & - 1
                \\
                \hline
                \m{S_{1, 2}} & \left(- 31 \pm i \sqrt{335}\right)/24
                \\
                \hline
                \m{S_{1, 3}} & - 1, - 1/3
                \\
                \hline
                \m{S_{2, 3}} & - 1, - 9/4
                \\
                \hline
            \end{tabular}
            \caption{Eigenvalues of the principal submatrices of $\jac$ for the 3-component Schnakenberg system, \cref{Schnakenberg}, when $(a, b) = (1, 0.5)$.}
            \label{tab:subeigval}
        \end{table}
        }
        The Jacobian matrix and all of its principal submatrices are stable. This means that a diagonal diffusion matrix is unable to produce diffusion-driven instabilities \cite{satnoianu,villar-sepulveda}. However, we can find a cross-diffusion matrix so that the system exhibits wave instabilities. As detailed in the proof of \cref{lemma:wave}, we can choose an unstable $2 \times 2$ matrix with two unstable eigenvalues to find a suitable diffusion matrix that destabilizes the system.

        In particular, we note that $\jac$ has three different eigenvalues, so we can find parameters $\alpha, \gamma, \ell_1, \ell_2, \ell_3 \in \mathbb R$ so that $\jac$ is similar to
        \begin{align*}
            \mbf C = \begin{pmatrix}
                197 & \alpha & 1
                \\
                \ell_1 & \ell_2 & \ell_3
                \\
                - 39072 & \gamma & - \frac{2353}{12}
            \end{pmatrix}.
        \end{align*}
        We choose $\alpha = 0$ and $\gamma = \frac{9}{4}$ and note that these two matrices are similar if and only if
        \begin{align*}
            \ell_1 = 38940, \qquad \ell_2 = - \frac{9}{2}, \qquad \ell_3 = \frac{1159}{6}.
        \end{align*}
        With these parameter values, we can solve the system necessary for similarity, i.e.
        \begin{align*}
            \m{Q} \m{C} = \jac \m{Q},
        \end{align*}
        to find that
        \begin{align*}
            \m{Q} = \begin{pmatrix}
                198 & 0 & 1
                \\
                0 & 1 & 1
                \\
                1 & 0 & 0
            \end{pmatrix},
        \end{align*}
        which fulfills $\det(\m{Q}) = - 1 \neq 0$.
        
        Note that the matrix $\m{C}$ has a principal submatrix $\m{S_{1, 3}}$ with eigenvalues $\frac{1}{24} \left(11 \pm \sqrt{255383} i\right)$, which are complex values with a positive real part, and hence our candidate submatrix to induce a wave instability.
        
        Now we pick the diagonal matrix
        \begin{align*}
            \Db = \begin{pmatrix}
                \delta & 0 & 0
                \\
                0 & 1 & 0
                \\
                0 & 0 & \delta
            \end{pmatrix},
        \end{align*}
        to select the unstable principal submatrix of $\m{C}$ we have designed via the parameter $\delta$, arriving at the diffusion matrix
        \begin{align}
            \D = \m{Q} \Db \m{Q}^{- 1} = \left(
                \begin{array}{ccc}
                     \delta & 0 & 0
                     \\
                     \delta - 1 & 1 & 198 (1 - \delta)
                     \\
                     0 & 0 & \delta
                \end{array}
            \right). \label{diffmatrixSchnakenberg}
        \end{align}
        With this, we will be able to exhibit a wave instability, provided that $\delta$ is small enough. In particular, if we let $\delta$ range from 0 to 1 in discrete steps, we obtain the dispersion relations shown in \cref{fig:disprelSchnakenberg}. Note that the value given by the green dashed line corresponds to the asymptotic value of $\Re(\lambda)$ as $\mu \to \infty$ for the degenerate system at $\delta = 0$, which corresponds to the real part of the eigenvalues of the unstable submatrix we designed in $\m{C}$. We plot the corresponding dispersion curve for $\delta = 0$ in orange, noting that it is not an admissible value of $\delta$, as the linear system is ill-posed if $\Re(\lambda) > 0$ for arbitrarily large $\mu$. While this particular value is inadmissible, what we construct in this way is a dispersion relation that has the correct properties for all $\delta > 0$. Note, in particular, that one can look at the curves between the red and the orange one (that is, $\delta \in (0, 0.842742)$), and see that any of these intermediate values of $\delta$ lead to the right form of a dispersion relation for a diffusion-driven instability. This is the key insight of \cref{lemma:positiveeig} on which the constructions are based.

        \begin{figure}
            \centering
            \includegraphics[width = \textwidth]{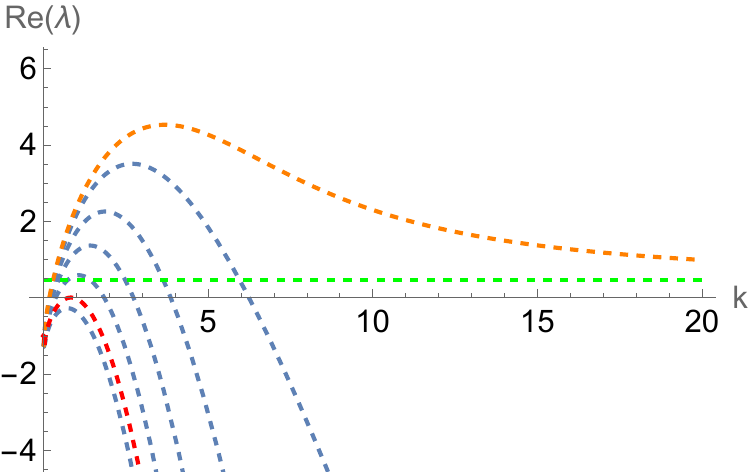}
            \caption{Dispersion relation of the 3-component Schnakenberg system, \cref{Schnakenberg}, for $(a, b) = (1, 0.5)$ and the diffusion matrix given by \cref{diffmatrixSchnakenberg}. Here, $\delta$ ranges from 0 (in orange) to 1 in discrete steps. The top (respectively, bottom) curve corresponds to $\delta = 0$ (respectively, $\delta = 1$). Moreover, the green dashed line corresponds to the horizontal line $\Re(\lambda) = 11/24$, whilst the red dashed line is a wave bifurcation curve found for $\delta \approx 0.842742$ where $\Re(\lambda) = 0$ for some $k > 0$.}
            \label{fig:disprelSchnakenberg}
        \end{figure}
        
        Hence, for sufficiently small $\delta > 0$, we anticipate that the full system exhibits a wave instability. To ensure that the pattern is visible, i.e., that there is a nearby stable patterned state, we compute the criticality of both bifurcations, the traveling wave and the standing wave, using the approach outlined in \cite{Villar-sepulveda-wave}. The bifurcation for both, travelling and standing waves turn\ak{s} out to be supercritical in the set of values $(b, \delta) \in [0, 1] \times [0, 1]$ on an infinite domain, indicating that they should be feasible to obtain via simulations. However, only the traveling wave turns out to be a stable pattern with periodic boundary conditions (see \cite{knobloch}), whilst the standing wave turns out to be stable in the presence of homogeneous Neumann boundary conditions. In particular, for $\delta = 0.8$ we set the length of the domain as $L = 2.38 \pi$ and integrated the system with periodic boundary conditions. The solution we obtained is shown in \cref{fig:Schnakenberg-pattern-tw}, where we see that a travelling wave emerges from a smaller amplitude standing wave transient. We also integrated the system with homogeneous Neumann boundary conditions to observe the solution shown in \cref{fig:Schnakenberg-pattern-sw}. 
        \begin{figure}
             \centering
             \includegraphics[width = \textwidth]{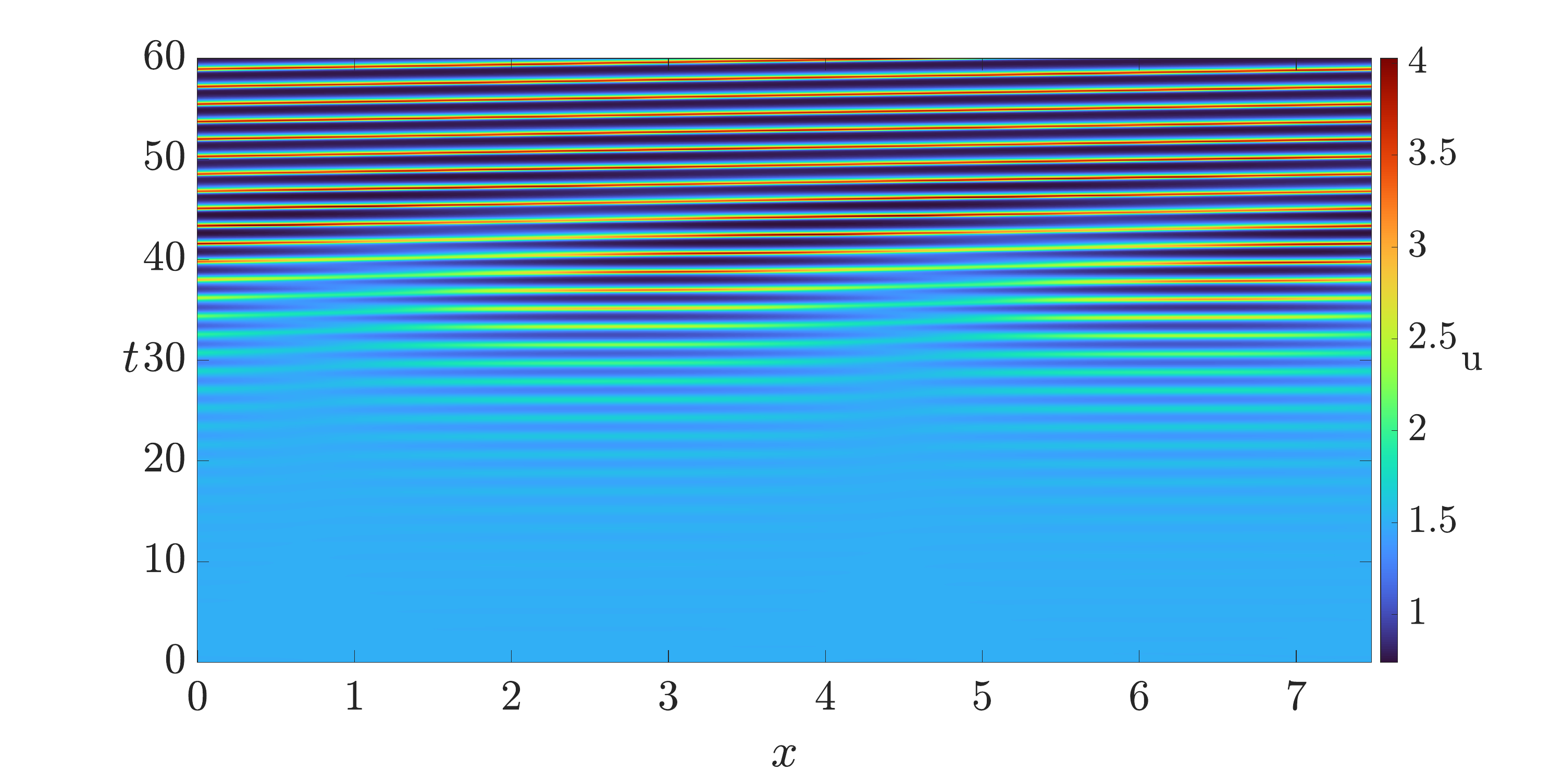}
             \caption{Solution $u$ of the 3-component Schnakenberg system, \cref{Schnakenberg}, with periodic boundary conditions, and with the diffusion matrix given by \cref{diffmatrixSchnakenberg} and $(a, b, \delta) = (1, 0.5, 0.8)$.}
            \label{fig:Schnakenberg-pattern-tw}
        \end{figure}

        \begin{figure}
             \centering
             \includegraphics[width = \textwidth]{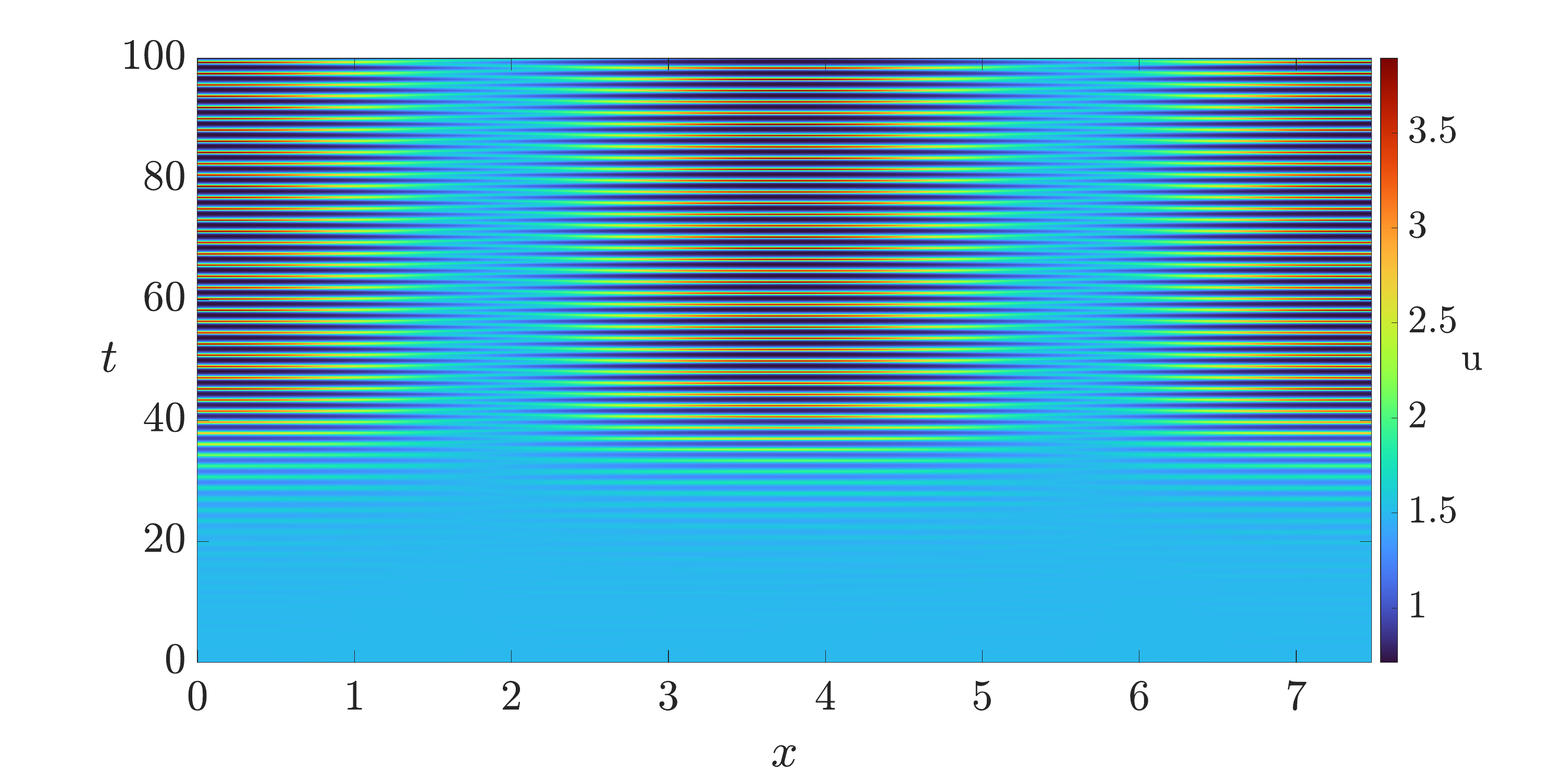}
             \caption{Solution $u$ of the 3-component Schnakenberg system, \cref{Schnakenberg}, with homogeneous Neumann boundary conditions, and with a diffusion matrix given by \cref{diffmatrixSchnakenberg} and $(a, b, \delta) = (1, 0.5, 0.8)$.}
            \label{fig:Schnakenberg-pattern-sw}
        \end{figure}

    \subsection{Hyperbolic reaction-diffusion system} \label{ex:fhn}
        We next provide an example based on a hyperbolic system. In \cite{ritchie2022turing}, the authors study a general class of second order in time cross-diffusion systems. Such systems have particular motivations coming from understanding the microscopic inertia responsible for preventing the infinite speed of propagation that occurs in usual diffusion equations. \ak{See \cite{mendez2010reaction} for an overview of hyperbolic and other formulations of reaction-transport processes aiming to overcome this difficulty, and some of their connections with thermodynamic principles.} They also arise as truncations or linearizations of some nonlocal-in-time systems, e.g.~those due to fixed or distributed time delays arising in gene expression \cite{sargood2022fixed}. As an example of applying our constructions to these systems, we consider the following system from \cite{ritchie2022turing}:
        \begin{align}
            \tau \frac{\partial^2 u}{\partial t^2} + \frac{\partial u}{\partial t} = d_{11} \frac{\partial^2 u}{\partial x^2} + d_{12} \frac{\partial^2 v}{\partial x^2} + f(u, v), \notag
            \\
            \tau \frac{\partial^2 v}{\partial t^2} + \frac{\partial v}{\partial t} = d_{21} \frac{\partial^2 u}{\partial x^2} + d_{22} \frac{\partial^2 v}{\partial x^2} + g(u, v), \label{hyperbolic}
        \end{align}
        with a particular choice of kinetics of FitzHugh-Nagumo type
        \begin{align*}
            f(u, v) &= u - a u^3 + v - b,
            \\
            g(u, v) &= b - u - 3 a v.
        \end{align*}
        This system turns out to be equivalent to a degenerate reaction-diffusion system given by
        \begin{align*}
            \frac{\partial u}{\partial t} &= w,
            \\
            \frac{\partial v}{\partial t} &= z,
            \\
            \frac{\partial w}{\partial t} &= \frac{1}{\tau} \left(u - a u^3 + v - b - w + d_{11} \frac{\partial^2 u}{\partial x^2} + d_{12} \frac{\partial^2 v}{\partial x^2}\right),
            \\
            \frac{\partial z}{\partial t} &= \frac{1}{\tau} \left(b - u - 3 a v - z + d_{21} \frac{\partial^2 u}{\partial x^2} + d_{22} \frac{\partial^2 v}{\partial x^2}\right).
        \end{align*}
        While the analysis in \cite{ritchie2022turing} can be used to determine if a given set of parameters admits Turing or wave instabilities, the approach we use below allows us to constructively build a particular set of diffusion parameters (for a slightly modified system). Importantly, the method here generalizes to systems of $n > 2$ hyperbolic equations, though it is easier to see in this simplified setting.
        
        The diffusion matrix is given by
        \begin{align*}
            \D = \begin{pmatrix}
                0 & 0 & 0 & 0
                \\
                0 & 0 & 0 & 0
                \\
                \frac{d_{11}}{\tau} & \frac{d_{12}}{\tau} & 0 & 0
                \\
                \frac{d_{21}}{\tau} & \frac{d_{22}}{\tau} & 0 & 0
            \end{pmatrix}.
        \end{align*}
        As this system involves non-diffusible elements (see e.g. \cite{korvasova2015investigating}), we cannot immediately apply the constructions developed throughout this article, as all the eigenvalues of $\hat D$ are equal to zero for any diffusion rates. However, if we introduce a small change to this matrix, corresponding to a kind of parabolic perturbation, then we will be able to analyze it as a reaction-cross-diffusion system to produce Turing or wave instabilities using the theory we have developed. In particular, if we consider
        \begin{align}
            \D = \begin{pmatrix}
                \sigma_1 & 0 & 0 & 0
                \\
                0 & \sigma_1 & 0 & 0
                \\
                \frac{d_{11}}{\tau} & \frac{d_{12}}{\tau} & \sigma_2 & 0
                \\
                \frac{d_{21}}{\tau} & \frac{d_{22}}{\tau} & 0 & \sigma_2
            \end{pmatrix}, \label{diffmatrixhyperbolic}
        \end{align}
        then
        \begin{align*}
            \D = \m{Q} \Db \m{Q}^{- 1},
        \end{align*}
        where
        \begin{align*}
            \m{Q} &= \left(
                \begin{array}{cccc}
                 \frac{d_{1 2} \sigma_1 \tau - d_{12} \sigma_2 \tau}{d_{12} d_{21} - d_{11} d_{22}} & \frac{d_{22} \left(\sigma_1 - \sigma_2\right) \tau}{d_{11} d_{22} - d_{12} d_{21}} & 0 & 0
                 \\[1ex]
                 \frac{d_{11} \left(\sigma_1 - \sigma_2\right) \tau}{d_{11} d_{22} - d_{12} d_{21}} & \frac{d_{21} \left(\sigma_1 - \sigma_2\right) \tau}{d_{12} d_{21} - d_{11} d_{22}} & 0 & 0
                 \\[1ex]
                 0 & 1 & 0 & 1
                 \\
                 1 & 0 & 1 & 0
                \end{array}
            \right),
            \\
            \Db &= \begin{pmatrix}
                \sigma_1 & 0 & 0 & 0
                \\
                0 & \sigma_1 & 0 & 0
                \\
                0 & 0 & \sigma_2 & 0
                \\
                0 & 0 & 0 & \sigma_2
            \end{pmatrix}.
        \end{align*}
        Now, if we consider the steady state given by
        \begin{align*}
            \mbf P = \left(u^*, v^*, 0, 0\right),
        \end{align*}
        where
        \begin{align*}
            u^* &= \frac{\sqrt[3]{2} M^2 + 6 a^3 - 2 a^2}{3 \cdot 2^{2/3} a^2 M},
            \\
            v^* &= \frac{-2^{2/3} M^2 - 6 \sqrt[3]{2} a^3 + 2 a^2 \left(3 b M + \sqrt[3]{2}\right)}{18 a^3 M},
            \\
            M &= \left(9 (1 - 3 a) a^4 b + \sqrt{N}\right)^{1/3},
            \\
            N &= (1 - 3 a)^2 a^6 \left(3 a \left(27 a b^2 - 4\right) + 4\right),
        \end{align*}
        and take, for example, the parameter values $(a, b, \tau) = (0.257, 0.98, 0.1)$, then
        \begin{align*}
            \jac = \left(
                \begin{array}{cccc}
                     0 & 0 & 1 & 0
                     \\
                     0 & 0 & 0 & 1
                     \\
                     6.30191 & 10 & -10 & 0
                     \\
                     -10. & -7.71 & 0 & -10
                \end{array}
            \right),
        \end{align*}
        which turns out to be stable. Furthermore, the submatrices $\m{S_{1, 2}}$ and $\m{S_{3, 4}}$ of $\m{Q}^{- 1} \jac \m{Q}$ turn out to be
        \begin{align*}
            \boldsymbol{\alpha_{1, 2}} = \begin{pmatrix}
                \frac{d_{22}}{\tau \left(\sigma_1 - \sigma_2\right)} & \frac{d_{21}}{\tau \left(\sigma_1 - \sigma_2\right)}
                 \\[1ex]
                 \frac{d_{12}}{\tau \left(\sigma_1 - \sigma_2\right)} & \frac{d_{11}}{\tau \left(\sigma_1 - \sigma_2\right)}
            \end{pmatrix}, \quad \text{and} \quad \boldsymbol{\alpha_{3, 4}} := \begin{pmatrix}
                - \frac{d_{22} + \sigma_1 - \sigma_2}{\tau \left(\sigma_1 - \sigma_2\right)} & - \frac{d_{21}}{\tau \left(\sigma_1 - \sigma_2\right)}
                \\
                - \frac{d_{12}}{\tau \left(\sigma_1 - \sigma_2\right)} & -\frac{d_{11} + \sigma_1 - \sigma_2}{\tau \left(\sigma_1 -\sigma_2\right)}
            \end{pmatrix},
        \end{align*}
        respectively. This implies that we can make any of these matrices unstable\ak{,} and that will let us generate diffusion-driven instabilities through $\sigma_1$ or $\sigma_2$, respectively. For concreteness, we set $\sigma_2 = 1$ and leave $\sigma_1$ as a free non-negative parameter.
        
        If we make $\boldsymbol{\alpha_{1, 2}}$ unstable, we can take $\sigma_1 = 0$ to generate a system with an instability and then increase that parameter a small amount (see \cref{lemma:positiveeig}). When setting $\sigma_1 = 0$, we get
        \begin{align*}
            \boldsymbol{\alpha_{1, 2}} = \left(
                \begin{array}{cc}
                 - \frac{d_{22}}{\tau \sigma_2} & - \frac{d_{21}}{\tau \sigma_2}
                 \\[1ex]
                 - \frac{d_{12}}{\tau \sigma_2} & - \frac{d_{11}}{\tau \sigma_2}
                \end{array}
            \right).
        \end{align*}
        The only thing we need to do is to find parameters so that $\boldsymbol{\alpha_{1, 2}}$ is unstable. In particular, if we set
        \begin{align}
            \left(d_{11}, d_{12}, d_{21}, d_{22}\right) = \left(- 1, - 1, 0, - 2\right), \label{diffparvalhyperbolic}
        \end{align}
        then $\boldsymbol{\alpha_{1, 2}}$ becomes
        \begin{align*}
            \begin{pmatrix}
                20 & 0
                \\
                10 & 10
            \end{pmatrix},
        \end{align*}
        which has positive real eigenvalues. We then take $\sigma_1 = \delta$ and vary this parameter to generate the dispersion relation shown in \cref{fig:disprelFHN}, where $\delta$ goes from 0 to 2 in uniform steps of size 0.4. This graph is particularly interesting as it lets us see that, depending on the value of $\delta$, we find an interaction between Turing and wave instabilities (the eigenvalue with the largest real part is complex with a non-zero imaginary part for low values of $k$ but it becomes real for larger wavenumbers). Note that, for each $\delta$ and $k$, we are plotting the two eigenvalues with the largest real parts, which each coincide with the complex-conjugate pairs for small $k$ but bifurcate for larger $k$ into two strictly real branches. As before, the orange curves correspond to the inadmissible value of $\delta = 0$, at which each curve tends to a different asymptotic eigenvalue shown in green lines. Again, we obtain the correct instability for those curves between the orange and red curves.
        
        In particular, if we set homogeneous Neumann boundary conditions, take $\delta = 1.3$ with a domain length $L = 6$, and integrate the system, we obtain the pattern shown in \cref{fig:hyperbolic-pattern}. On the other hand, if we set periodic boundary conditions, with the same value of $\delta$, but a domain length $L = 30$ and integrate the system, we obtain the pattern shown in \cref{fig:hyperbolic-pattern-periodic}.
        \begin{figure}
            \centering
            \includegraphics[width = \textwidth]{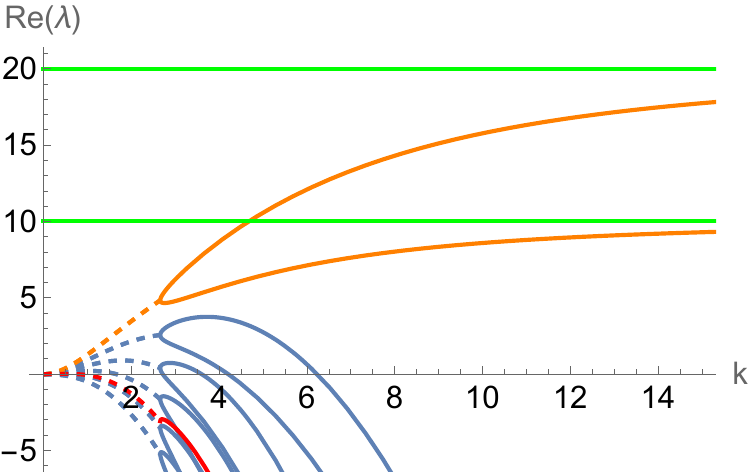}
            \caption{Dispersion relations of the hyperbolic system, \cref{hyperbolic}, when $(a, b, \tau) = (0.257, 0.98, 0.1)$, with a diffusion matrix given by \cref{diffmatrixhyperbolic}, and the diffusion parameters provided in \cref{diffparvalhyperbolic}, together with $\left(\sigma_1, \sigma_2\right) = (\delta, 1)$. Here, $\delta$ ranges from 0 (in orange) to 2 in uniform steps of size 0.4. The top (respectively, bottom) curve corresponds to $\delta = 0$ (respectively, $\delta = 2$). Moreover, the green lines correspond to the lines $\Re(\lambda) = 10$, and $\Re(\lambda) = 20$, whilst the red line is a wave bifurcation curve found for $\delta \approx 1.507568$.}
            \label{fig:disprelFHN}
        \end{figure}
        \begin{figure}
             \centering
             \includegraphics[width = \textwidth]{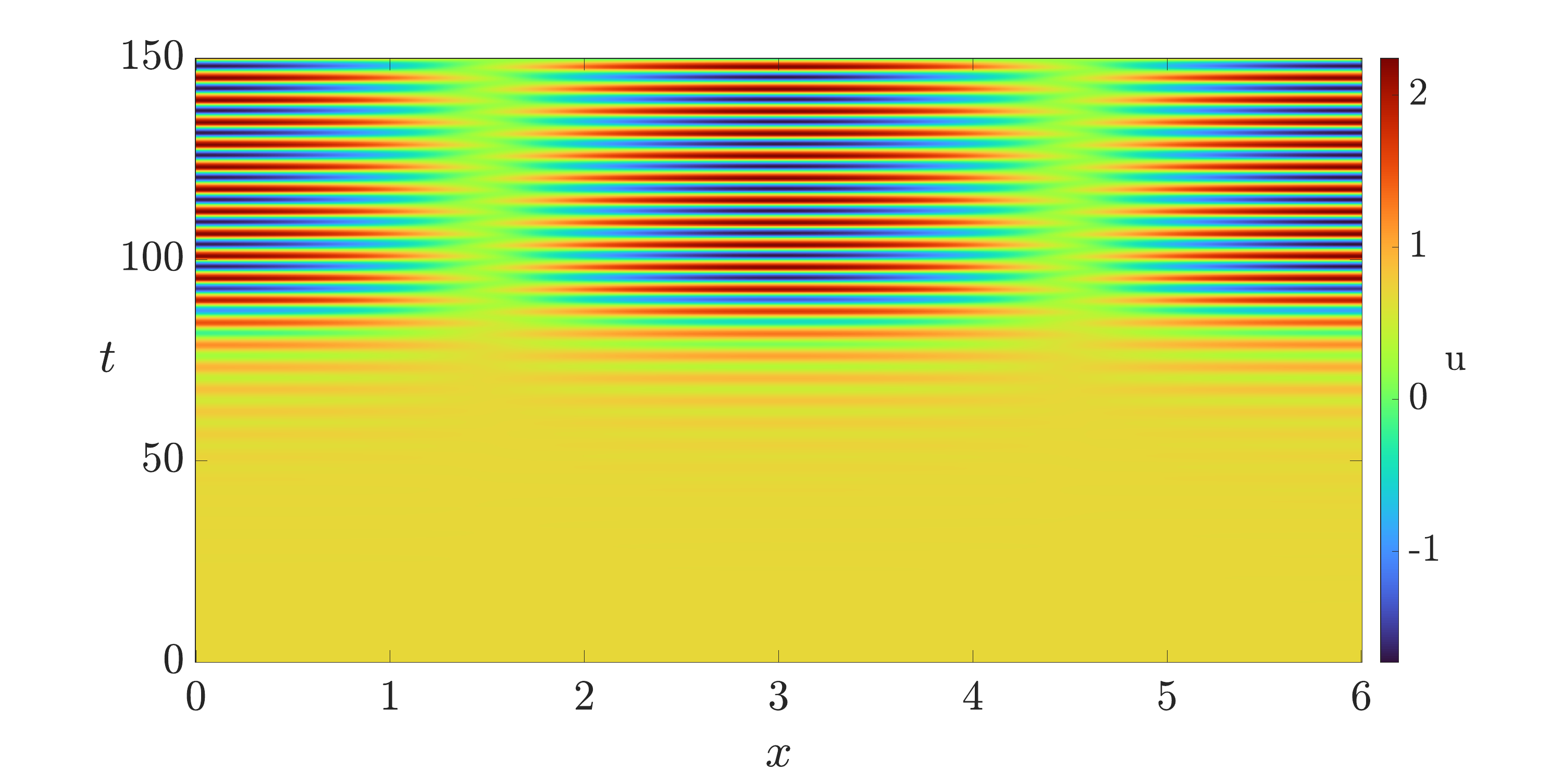}
             \caption{Solution $u$ of the 2-component hyperbolic system, \cref{hyperbolic}, with homogeneous Neumann boundary conditions, a diffusion matrix given by \cref{diffmatrixhyperbolic}, $(a, b, \tau) = (0.257, 0.98, 0.1)$, and the diffusion parameters provided in \cref{diffparvalhyperbolic}, together with $\left(\sigma_1, \sigma_2\right) = (1.3, 1)$.}
             \label{fig:hyperbolic-pattern}
        \end{figure}

        \begin{figure}
             \centering
             \includegraphics[width = \textwidth]{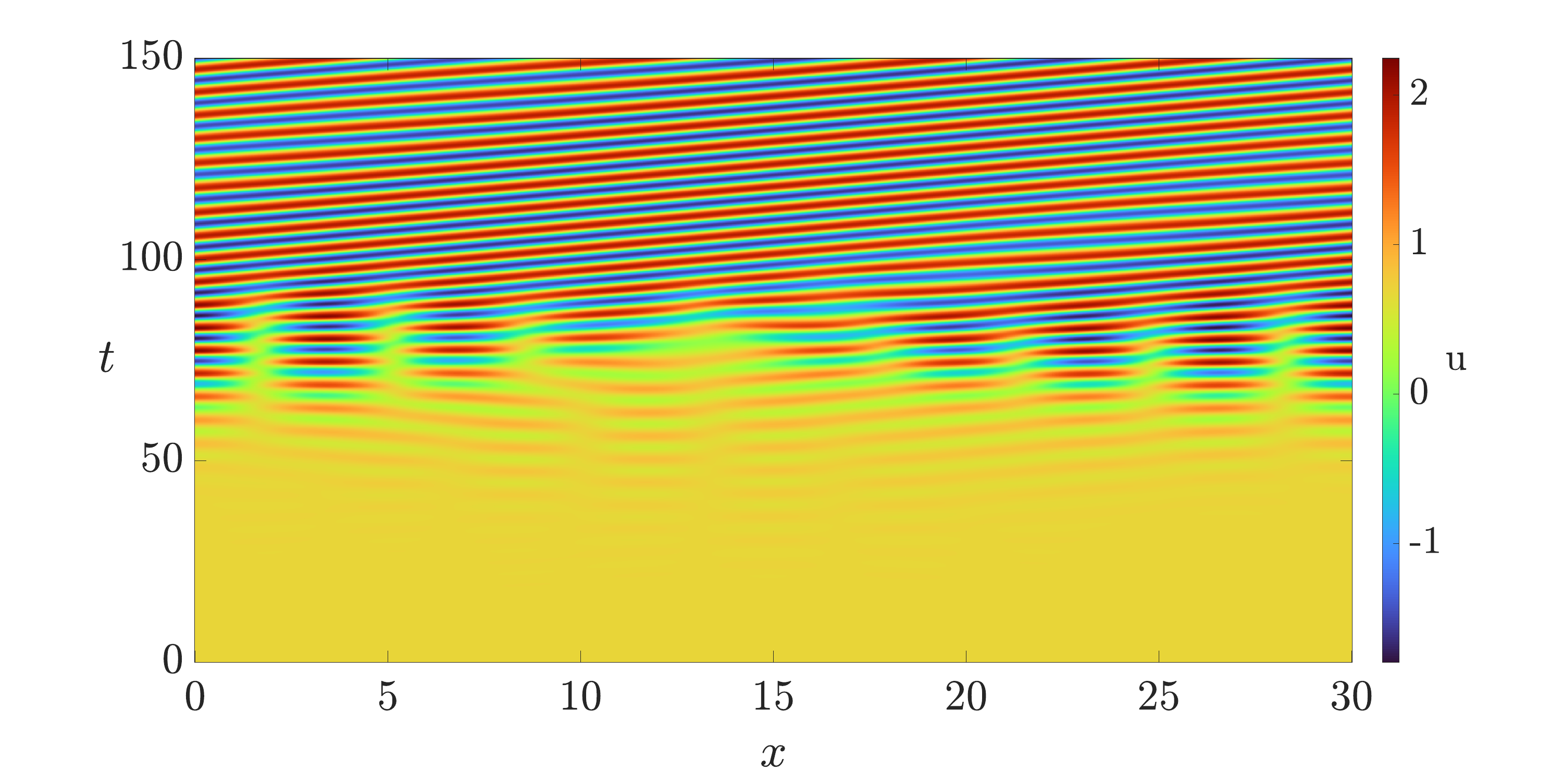}
             \caption{Solution $u$ of  the 2-component hyperbolic system, \cref{hyperbolic}, with periodic boundary conditions, a diffusion matrix given by \cref{diffmatrixhyperbolic}, $(a, b, \tau) = (0.257, 0.98, 0.1)$, and the diffusion parameters provided in \cref{diffparvalhyperbolic}, together with $\left(\sigma_1, \sigma_2\right) = (1.3, 1)$.}
             \label{fig:hyperbolic-pattern-periodic}
        \end{figure}

        The solution $u$ in \cref{fig:hyperbolic-pattern} takes negative values. This is not necessarily unphysical due to the nonlinearities being chosen to match a version of the FitzHugh-Nagumo model, where $u$ represents a membrane voltage, and hence its sign is not constrained. In other cases, however, one has to be careful using both hyperbolic terms or linear cross-diffusion to maintain positive solutions \cite{mendez2010reaction}, as we will see in the next example. 

        In this example, we used \cref{lemma:positiveeig,lemma:n=2} to generate diffusion-driven instabilities in a reaction-diffusion system arising from a hyperbolic equation. This example is full of interesting features as some interactions between Turing and wave instabilities can be found. Some of these interactions can be explained by the use of \cite[Theorem 4]{villar-sepulveda} for the matrix $\m{Q}^{- 1} \jac \m{Q}$ instead of simply $\jac$.

    \subsection{3-component malaria model}
        In \cite{villar-sepulveda}, the authors provided an example of a 3-component malaria model that was shown to be unable to exhibit Turing or wave instabilities for any diagonal diffusion matrix. We now show that we can generate diffusion-driven instabilities in this model if we consider chemotactic fluxes among the populations. In particular, we consider said system as
        \begin{align} \label{malaria}
            \partial_t \u = \mbf f(\u) + \partial_x \left(\m{D}(\u) \partial_x \u\right),
        \end{align}
        where
        \begin{align}
            \u = (H, I, P)^\intercal, \quad \mbf f(\u) = \begin{pmatrix}
                (b_H - d_H) H - c P H + r I
                \\
                - d_H I + c P H - r I
                \\
                - d_M P + b (Q - P) I
            \end{pmatrix}, 
        \end{align}
        and the nonlinear diffusion tensor is given by
        \begin{align}
            \m{D}(\u) = \begin{pmatrix}
                D_{1, 1} & \chi\left(\frac{H}{H^*} \right) D_{1, 2} & \chi\left(\frac{H}{H^*}\right) D_{1, 3}
                \\
                \chi\left(\frac{I}{I^*}\right) D_{2, 1} & D_{2, 2} & 0
                \\
                0 & 0 & D_{3, 3} 
            \end{pmatrix}, \quad \chi(V) = \frac{2 V}{1 + V^2},
        \end{align}
        The variables $H$ and $I$ represent the populations of healthy and infected humans, respectively, whilst $P$ is the population of infected mosquitoes. Moreover, $b_H$ represents the birth rate of healthy humans, $d_H$ the death rate of humans, $c$ the infection rate between humans, $r$ the human recovery rate, $d_M$ the death rate of infected mosquitoes, and $b$ is the infection rate between mosquitoes (see \cite{alonso} for discussion of substantially more complicated compartment-based models of malaria transmission dynamics, from which the present model is a simplified albeit spatially-extended variant). Note that healthy mosquitoes do not appear explicitly in this model, as $Q$ represents the total (fixed) mosquito population, so $Q - P$ is the population of healthy mosquitoes.
        
        We interpret the cross-diffusion terms as chemotactic fluxes \cite{hillen2009user} with saturating sensitivities corresponding to volume-filling effects \cite{bubba2020discrete}. This represents populations moving along signalling gradients of the other populations. We have assumed several zeros in the diffusion matrix as we do not expect mosquitoes' spatial movements to be largely influenced directly by human populations. To fix ideas, we will focus on generating Turing instabilities by varying the diffusion matrix, with a constrained diffusion tensor having the form given by \cref{diffmatrixmalaria}. The variables $H^*$ and $I^*$ represent homogeneous steady-state values of these parameters. Due to the choice of the sensitivity function $\chi$, we have that $\chi(1) = 1$. Hence, after linearizing around such a homogeneous steady state, we have the linear diffusion tensor:
        \begin{align}
            \D = \begin{pmatrix}
                D_{1, 1} & D_{1, 2} &  D_{1, 3}
                \\
                 D_{2, 1} & D_{2, 2} & 0
                \\
                0 & 0 & D_{3, 3} 
            \end{pmatrix}.\label{diffmatrixmalaria}
        \end{align} 
        

        To preserve the physical interpretations of this model, we explicitly state that all the variables and parameters are non-negative, $b_H > d_H > 0$, and $P \leq Q$. We also assume the diffusion matrix to have positive real eigenvalues, so we need $D_{3, 3} > 0$, together with
        \begin{align*}
            D_{1, 1} + D_{2, 2} &> 0,
            \\
            D_{1, 1} D_{2, 2} - D_{1, 2} D_{2, 1} &> 0,
            \\
            \left(D_{1, 1} + D_{2, 2}\right)^2 - 4 \left(D_{1, 1} D_{2, 2} - D_{1, 2} D_{2, 1}\right) &> 0.
        \end{align*} 
        This system has only two equilibria, $\mathbf 0 = (0, 0, 0)$ and $\mbf P = \left(H^*, I^*, P^*\right)$, where
        \begin{align*}
        	H^* &= \frac{d_H d_M \left(d_H + r\right)}{b  \left(d_H \left(c Q + d_H + r\right) - b_H \left(d_H + r\right)\right)},
        	\\
        	I^* &= \frac{d_M \left(b_H - d_H\right) \left(d_H + r\right)}{b \left(d_H \left(c Q + d_H + r\right) - b_H \left(d_H + r\right)\right)},
        	\\
        	P^* &= \frac{\left(b_H - d_H\right) \left(d_H + r\right)}{c d_H}.
        \end{align*}
        We assume $d_H \left(c Q + d_H + r\right) - b_H \left(d_H + r\right) > 0$ to ensure the feasibility of the endemic (positive) steady state. Note that a Turing instability around the origin will not be biologically relevant since the existence of such patterning would imply that some of the variables become negative.
    
        If we set the following parameter values:
        \begin{align}
            \left(b_H, d_H, c, r, d_M, b, Q\right) = (1, 0.1, 0.25, 0.5, 0.3, 0.1, 100), \label{parval_malaria}
        \end{align}
        then the Jacobian matrix of the system at $\mbf P$ in the absence of diffusion becomes
        \begin{align*}
            \jac = \left(
                \begin{array}{ccc}
                     -4.5 & 0.5 & -0.0229592
                     \\
                     5.4 & -0.6 & 0.0229592
                     \\
                     0 & 7.84 & -0.382653
                \end{array}
            \right),
        \end{align*}
        which has eigenvalues $(- 5.14441, - 0.169119 \pm 0.0537504 i)$, meaning that it is stable.
       
        Now, note that the diffusion matrix has the diagonal form
        \begin{align*}
            \D = \m{R} \Db \m{R}^{- 1},
        \end{align*}
        where
        \begin{align*}
            \m{R} &= \left(
                \begin{array}{ccc}
                     R_{1, 1} & R_{1, 2} & R_{1, 3}
                     \\
                     1 & 1 & R_{2, 3}
                     \\
                     0 & 0 & 1
                \end{array}
            \right),
            \\[1.5ex]
            \Db &= \left(
                \begin{array}{ccc}
                     \sigma_1 & 0 & 0
                     \\
                     0 & \sigma_2 & 0
                     \\
                     0 & 0 & \sigma_3
                \end{array}
            \right),
            \\
            R_{1, 1} &= - \frac{\sqrt{\left(D_{1, 1} - D_{2, 2}\right)^2 + 4 D_{1, 2} D_{2, 1}} - D_{1, 1} + D_{2, 2}}{2 D_{2, 1}},
            \\
            R_{1, 2} &= \frac{\sqrt{\left(D_{1, 1} - D_{2, 2}\right)^2 + 4 D_{1, 2} D_{2, 1}} + D_{1, 1} - D_{2, 2}}{2 D_{2, 1}},
            \\
            R_{1, 3} &= \frac{D_{1, 3} \left(D_{3, 3} - D_{2, 2}\right)}{\left(D_{1, 1} - D_{3, 3}\right) \left(D_{2, 2} - D_{3, 3}\right) - D_{1, 2} D_{2, 1}},
            \\
            R_{2, 3} &= \frac{D_{1, 3} D_{2, 1}}{\left(D_{1, 1} - D_{3, 3}\right) \left(D_{2, 2} - D_{3, 3}\right) - D_{1, 2} D_{2, 1}},
            \\
            \sigma_1 &= \frac{1}{2} \left(D_{1, 1} + D_{2, 2} - \sqrt{\left(D_{1, 1} - D_{2, 2}\right)^2 + 4 D_{1, 2} D_{2, 1}}\right),
            \\
            \sigma_2 &= \frac{1}{2} \left(D_{1, 1} + D_{2, 2} + \sqrt{\left(D_{1, 1} - D_{2, 2}\right)^2 + 4 D_{1, 2} D_{2, 1}}\right),
            \\
            \sigma_3 &= D_{3, 3}.
        \end{align*}
        While this lets us compute $\m{R}^{- 1} \jac \m{R}$ explicitly, we omit its full expression as it will not provide any obvious insights. What we do care about is the set of elements along its diagonal. The simplest such element is $M_3 = \left(\m{R}^{- 1} \jac \m{R}\right)_{3, 3}$ which is given by
        \begin{align*}
            M_3 &= \frac{7.84 D_{1, 3} D_{2, 1}}{D_{1, 1} D_{2, 2} -D_{1, 2} D_{2, 1}} - 0.382653,
        \end{align*}
        when $D_{3, 3} = 0$.
        In particular, if we set the diffusion rates as
        \begin{align}
            \left(D_{1, 1}, D_{1, 2}, D_{1, 3}, D_{2, 1}, D_{2, 2}, D_{3, 3}\right) = (1, 0.5, 0.307225, 0.870348, 1, \delta), \label{diffparval_malaria}
        \end{align}
        we find that $M_3 = 3.32887$ when $\delta = 0$, so this can act as a principal unstable submatrix. Moreover, when we let $\delta$ range from 0 to 0.1 in discrete steps, we obtain the dispersion relations shown in \cref{fig:disprelmalaria}.
        \begin{figure}
            \centering
            \includegraphics[width=\textwidth]{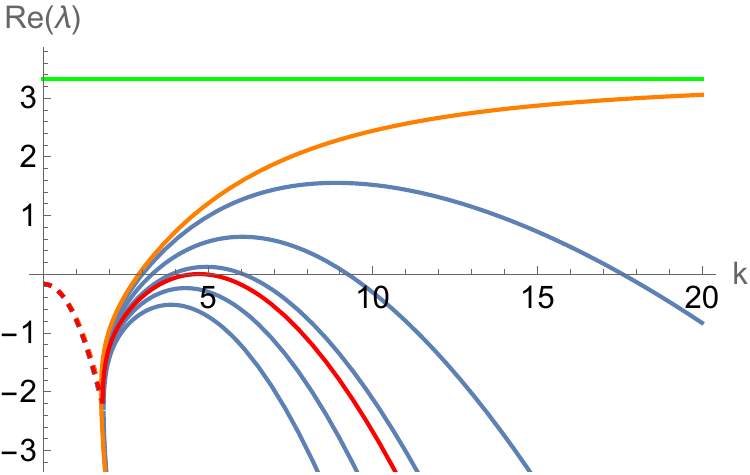}
            \caption{Dispersion relations of \cref{malaria} for the parameter values shown in \cref{parval_malaria,diffparval_malaria}. Here, $\delta$ ranges from 0 (in orange) to 0.1 in discrete steps. The top (respectively, bottom) curve corresponds to $\delta = 0$ (respectively, $\delta = 0.1$). Moreover, the green line corresponds to the horizontal line $\Re(\lambda) = 3.32887$, whilst the red line is a Turing bifurcation curve found for $\delta \approx 0.0561876$.}
            \label{fig:disprelmalaria}
        \end{figure}
        This implies that we have created a diffusion-driven instability in this model for different values of $\delta > 0$.
        
        \begin{figure}
            \centering
            \includegraphics[width = 0.27\textwidth]{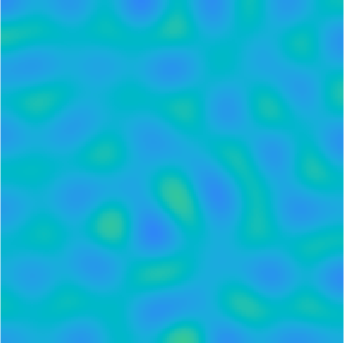} \hspace{0.1cm}     \includegraphics[width = 0.27\textwidth]{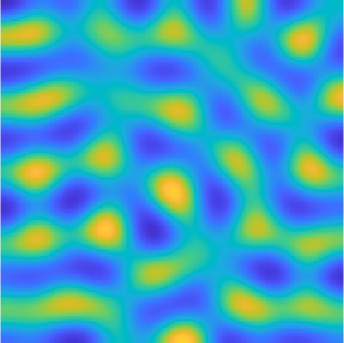} \hspace{0.1cm}
            \includegraphics[width = 0.27\textwidth]{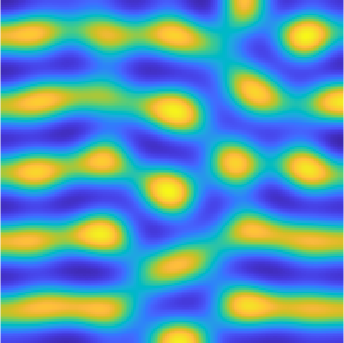} \hspace{0.1cm}
            \includegraphics[width = 0.065\textwidth]{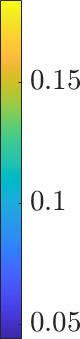}
            \\
            \hspace{-1cm} $t = 20$ \hspace{2cm} $t = 50$ \hspace{2cm} $t = 100$
            \caption{Three time points of the $H$-component of the solution of the 3-component malaria model, \cref{malaria}, solved in a square domain of side length $L = 6$ with homogeneous Neumann boundary conditions. The diffusion matrix is given by \cref{diffmatrixmalaria} with the kinetic parameters and diffusion rates as in \cref{parval_malaria,diffparval_malaria}, respectively. Solutions in $P$ and $I$ are out-of-phase with these patterns, but otherwise qualitatively similar.}
            \label{fig:malariapattern}
        \end{figure}

        We show examples of two-dimensional plots obtained in \cref{fig:malariapattern}, where the small perturbation of the homogeneous equilibrium relaxes to a stationary pattern over time. We remark that the choice of nonlinear diffusion here was motivated in part because linear cross-diffusion terms will not ensure that the variables remain positive, breaking both the boundedness of the solutions due to the quadratic nonlinearities and the physical interpretations of the state variables as populations. Nevertheless, the theory developed here is immediately applicable to suitable parameterized forms of nonlinear diffusion as shown in this example.
        

    \subsection{Keller-Segel model}
        One major motivation for cross-diffusion systems is the well-known Keller-Segel model for chemotaxis (see \cite{horstmann20031970, Keller-Segel}, and references therein). Here, we provide an example showing how we can use \cref{th:GIORNO-GIOVANNA} to generate diffusion-driven instabilities. We consider a three-species formulation (again to allow for wave instabilities) given by
        \begin{align}
            \partial_t \u = \mbf f(\mbf u) + \nabla \cdot (\mbf D \nabla \u), \label{KS}
        \end{align}
        where
        \begin{align}
            \mbf f(\u) = \begin{pmatrix}
                b u (1 - u)
                \\
                c u + e v + f w - \varepsilon v^3
                \\
                g u + h v + j w - \varepsilon w^3
            \end{pmatrix}, \quad \text{and} \quad \mbf D = \begin{pmatrix}
                d & - \frac{2 a u}{1 + u^2} & 0
                \\
                0 & 1 & 0
                \\
                0 & 0 & 1
            \end{pmatrix}, \label{KS-terms}
        \end{align}
        with $b, c, e, f, g, h, j \in \mathbb R$ being parameters of the reaction of the system, $a, d > 0$ being associated with the diffusion of the species in the system, and $0 < \varepsilon \ll 1$. In this example, we will assume that $d$ is fixed and we need to construct $\jac$ to obtain Turing and wave instabilities with a constrained Jacobian matrix having the form given by \cref{KS-diffmatrix}. We remark that the cubic nonlinearities in the equations for $v$ and $w$ are used solely to bound solutions, with linear interactions chosen for simplicity, but one can view the forthcoming linear analysis as applicable to other 3-species variations of these reaction kinetics.
        
        If we focus on the homogeneous steady state given by $\mbf P = (1, 0, 0)^\intercal$, then
        \begin{align}
            \jac = \begin{pmatrix}
                - b & 0 & 0
                \\
                c & e & f
                \\
                g & h & j
            \end{pmatrix}, \quad \text{and} \quad \D = \begin{pmatrix}
                d & - a & 0
                \\
                 0 & 1 & 0
                 \\
                 0 & 0 & 1
            \end{pmatrix}. \label{KS-diffmatrix}
        \end{align}        
        Here, we consider $d \neq 1$ to ensure that $\D$ is diagonalizable. With this, we have
        \begin{align*}
            \D =  \m{Q} \Db \m{Q}^{- 1},
        \end{align*}
        where
        \begin{align*}
            \m{Q} = \begin{pmatrix}
                0 && \frac{a}{d - 1} && 1
                \\[2ex]
                0 && 1 && 0
                \\
                1 && 0 && 0    
            \end{pmatrix}, \qquad \Db = \diag\left\{1, 1, d\right\}.
        \end{align*}        
        Therefore, we have that
        \begin{align*}
            \m{Q}^{- 1} \jac \m{Q} = \begin{pmatrix}
                j && \frac{a g}{d - 1} + h && g
                \\[2ex]
                f && \frac{a c}{d - 1} + e && c
                \\[2ex]
                - \frac{a f}{d - 1} && - \frac{a (a c + b d - b + (d - 1) e)}{(d - 1)^2} && - b - \dfrac{a c}{d - 1}
            \end{pmatrix},
        \end{align*}
        which has the following eigenvalues:
        \begin{align*}
            \lambda_1 &= - b,
            \\
            \lambda_\pm &= \frac{e + j \pm \sqrt{(e + i)^2 - 4 (e i - f h)}}{2}.
        \end{align*}
        This entails that, to make this matrix stable, we need to have $b > 0$, $e + j < 0$, and $e j - f h > 0$. Now, we fix $a = 3$. With this, we have two cases:
        
        \subsubsection*{The case $d < 1$:} In this case, following \cref{lemma:construction}, as there is only one entry in $\Db$ with the lowest value located in the lower right corner, we have that the eigenvalues of $\linop$ will be led by the bottom right entry of $\m{Q}^{- 1} \jac \m{Q}$ when said entry \ak{is} large. For example, if we consider $d = 0.9$, and consider the parameters
        \begin{align}
            (b, c, e, f, g, h, j) = (1, 36.7, - 32, - 1, - 1.5, - 1618, - 64), \label{KS-parvaldl1}
        \end{align}
        we have that the bottom right element of $\m{Q}^{- 1} \jac \m{Q}$ equals $1100 > 0$. This gives rise to the dispersion relation shown in \cref{fig:KS-dl1}. This shows that we have generated a diffusion-driven instability for the parameter values chosen. Furthermore, when $k = 5$, we have that $\lambda_{\max} = 20.4439 \in \mathbb R$, which implies that we are seeing a Turing instability.

        Furthermore, to check our prediction, we integrated the system numerically and obtained the solution shown in \cref{fig:KS-Turing-localized}. The solution initially forms a series of spikes that merge, eventually tending towards peaks commonly seen in  Keller-Segel type models more generally \cite{horstmann20031970,painter2011spatio}.
        \begin{figure}
            \centering
            \includegraphics[width = \textwidth]{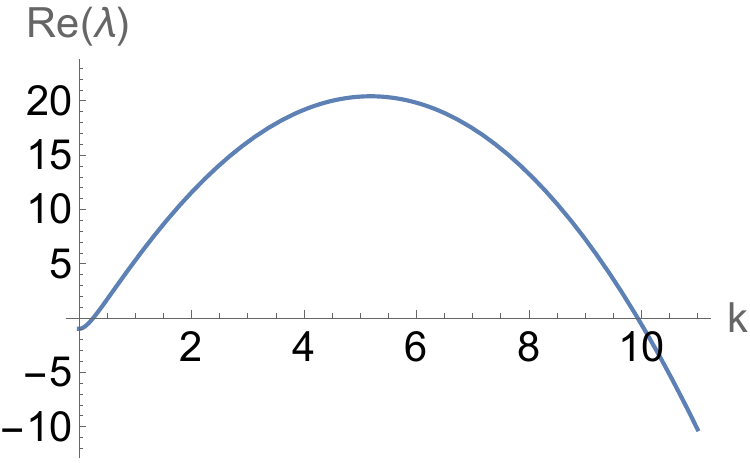}
            \caption{Dispersion relation for the Keller-Segel model, \cref{KS}, with \eqref{KS-terms}, for $a = 3$ and the parameter values as in \cref{KS-parvaldl1}, with $d = 0.9$.}
            \label{fig:KS-dl1}
        \end{figure}
        \begin{figure}
            \centering
            \includegraphics[width = \textwidth]{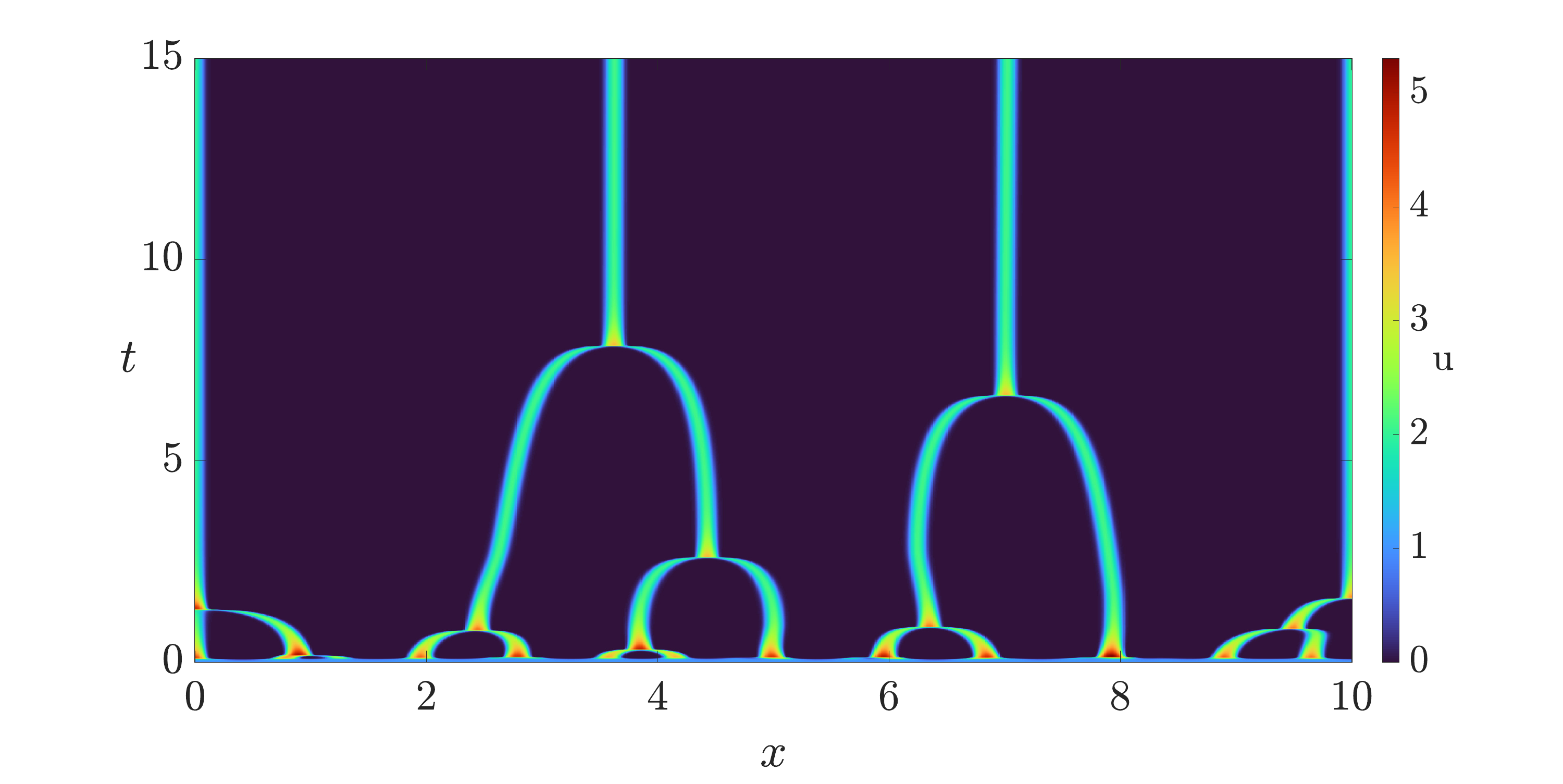}
            \caption{Solution $u$ of the Keller-Segel model, \cref{KS}, with \eqref{KS-terms}, when the parameters are fixed as in \cref{KS-parvaldl1} with $a = 3$ and $d = 0.9$.}
            \label{fig:KS-Turing-localized}
        \end{figure}
        
        \subsubsection*{The case $d > 1$:} In this case, once again following \cref{lemma:construction}, as there are two entries in $\Db$ with the lowest value, we can make the 1, 2-submatrix of $\m{Q}^{- 1} \jac \m{Q}$ unstable on a large scale to produce wave instabilities. In particular, if we consider $d = 1.1$ and
        \begin{align}
            (b, c, e, f, g, h, j) = (8, 33.75, - 512.5, - 500, - 0.41922, 512.577, 500), \label{KS-parval-dg1}
        \end{align}
        we get that this principal submatrix will be given by
        \begin{align*}
            \begin{pmatrix}
                500 & 500
                \\
                - 500 & 500
            \end{pmatrix},
        \end{align*}
        which is unstable with two complex conjugate eigenvalues. This gives rise to the dispersion relation shown in \cref{fig:KS-dg1}, which shows that we have generated a diffusion-driven instability for the parameter values chosen. Furthermore, when $k = 4$, we have that $\lambda_{\max} = 26.1765 \pm 75.7567 i \in \mathbb C$, which implies that we are seeing a wave instability.

        Once again, to check the outcome of the design process we have carried out, we integrated the system with homogeneous Neumann boundary conditions for these parameter values. In \cref{fig:KS-sw}, we show a solution leading to spatiotemporal oscillations, which is an outcome one expects from a wave instability \cite{Villar-sepulveda-wave}.
        \begin{figure}
            \centering
            \includegraphics[width = \textwidth]{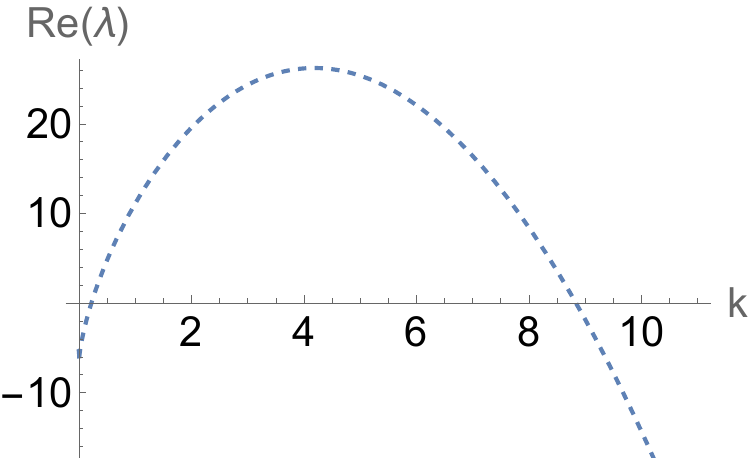}
            \caption{Dispersion relation for the Keller-Segel model, \cref{KS}, with \eqref{KS-terms}, for $a = 3$ and the parameter values as in \cref{KS-parval-dg1}, with $d = 1.1$.}
            \label{fig:KS-dg1}
        \end{figure}
        \begin{figure}
            \centering
            \includegraphics[width = \textwidth]{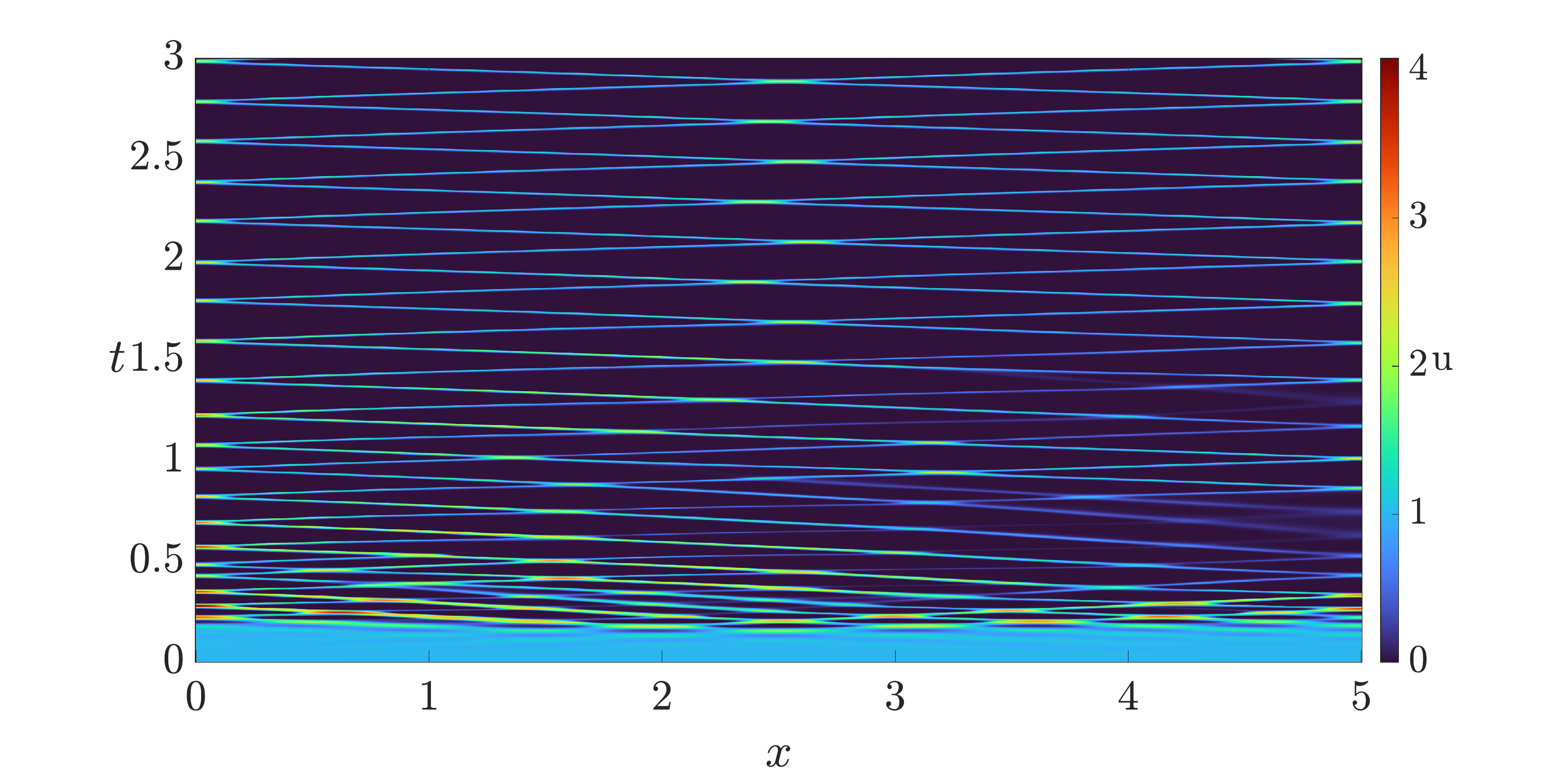}
            \caption{Solution $u$ of the Keller-Segel model, \cref{KS}, with \eqref{KS-terms} using homogeneous Neumann boundary conditions, when the parameters are fixed as in \cref{KS-parval-dg1} with $a = 3$ and $d = 1.1$.}
            \label{fig:KS-sw}
        \end{figure}
 
        In many instances of Keller-Segel-type models, the value of $d$ is equal to 1. That case is degenerate because the leading submatrix becomes non-diagonalizable for that value of $d$. In any case, just by continuity, one can continue the bifurcation using numerical continuation software and still get diffusion-driven instabilities when $d = 1$.


\section{Discussion}\label{sec:conclusions}
    We have developed an approach to designing reaction-cross-diffusion systems exhibiting Turing or wave instabilities for systems involving several components. We showed how to constructively build unstable principal submatrices which essentially encode these instabilities, and place these within the design constraints of a given system. The assumptions in \cref{th:main,th:GIORNO-GIOVANNA} are rather mild, suggesting that this process can be employed for large classes of systems, leaving flexibility in the constructions to allow for additional design constraints if they arise. This demonstrates some of the freedom gained in allowing non-diagonal diffusion tensors, or in having complete control over linearized kinetics. Finally, we showed these ideas in four distinct models with a range of linear and nonlinear transport, highlighting further issues of model construction that may arise, such as in preventing blowups or ensuring the positivity of solutions when this is required. We remark that all but the first example demonstrated employing these ideas in cases where we did not allow complete freedom over either matrix in the linearization, but instead had to work in a more constrained setting. \ak{An important thermodynamic consideration is to restrict to diffusion matrices $\D$ that are symmetric, as these can more easily be justified from the thermodynamic perspective \cite{klika2024onsager}. In such cases, Theorem \ref{th:GIORNO-GIOVANNA} can be applied directly, although the constructions used in proving Theorem \ref{th:main} would require significant changes that we leave for future work.}

    There are numerous future directions from the work carried out here. One could try to generalize \cref{th:GIORNO-GIOVANNA} by assuming that the diffusion matrix $\D$ admits complex eigenvalues. While such matrices are easier to work with in some sense (the set of diagonalizable matrices is dense in the set of invertible matrices), complex eigenvalues would require more care in carrying out the procedure given in \cref{lemma:positiveeig}, and hence may require entirely different constructions to arrive at something like \cref{th:GIORNO-GIOVANNA}. In a different direction, one could consider evolution equations with nonlocal spatial operators. Such integro-differential equations show stark differences with their classical local counterparts, allowing for spatiotemporal pattern formation in scalar models \cite{gourley2001spatio}, and showing dependence on the geometry not exhibited by local models even in the linear stability theory \cite{jewell2023patterning}. Such cases are more intricate as the symbol of the operator (i.e. the function of the wavenumber one obtains from the linearization in spectral space) is no longer a polynomial in general. Thus, developing an analogue of the procedure presented here would be substantially more involved. In contrast, the results we have presented essentially apply even to reaction-cross-diffusion systems posed on general Riemannian manifolds or even networks, modulo finite-size effects \cite[Section 3]{krause2021modern}.


    

    Although the linear theory developed is useful, Turing and wave instabilities do not ensure on their own that one will always be able to find specific kinds of solutions emerging from the instability of a spatially homogeneous equilibrium. We focused on the instability of a single homogeneous steady state, but realistic systems may admit more than one such state, leading to cases where the linear theory fails to capture the emergence of patterning due to other attracting states \cite{Andrewturing}. As we showed in the first example, calculations of the criticality of the bifurcations can be invaluable to ensure that a nearby stable pattern exists in the case of a supercritical bifurcation. In contrast, a subcritical bifurcation is a situation where the linear theory is not always a good indicator of subsequent emergent behaviour \cite{FahadWoods,villar2023degenerate,burke2007homoclinic,Villar-sepulveda-wave}. We note, in particular, that the two-species Keller-Segel model with logistic growth can exhibit spatiotemporal chaos after a subcritical Turing instability, despite such two-species systems not exhibiting wave instabilities or even Hopf bifurcations \ak{of homogeneous states} \cite{painter2011spatio} \ak{(though see \cite{ei2014spatio, kong2024existence} for cases where patterned states undergo Hopf bifurcations leading to these spatiotemporal behaviours)}. Hence, there is substantial work to be done in building design principles \ak{that} extend to the weakly nonlinear setting, and even developing a language of constraints for nonlinearities. \ak{Related to this limitation, the assumption of homogeneous Neumann or periodic conditions precludes the interplay of more complex boundary conditions which can modify patterns and even initiate patterning outside of the Turing space determined here via linear theory \cite{maini1997boundary, krause2021isolating}.}

    In a different direction, we note that, depending on the purpose of the mathematical model, having too much freedom in the design can lead to unsatisfying theories. Given full control over the nonlinearities, one can design even two-component reaction-diffusion systems that exhibit a wide variety of pre-defined behaviours \cite{woolley2021bespoke,woolley2025bespoke}, with essentially the freedom to match any kind of patterned state. Given such freedom, one must impose constraints that involve both linear and nonlinear aspects of the system under study to ensure that the choices made make sense in the modelling context. Similarly, while the results we have presented are general and powerful, one can argue that they provide too much freedom. In particular, there are few systems where one has complete control over either transport or reactions. Nevertheless, experimental protocols are increasingly providing more ways to create and tune a number of systems to exhibit patterning \cite{horvath2009experimental,karig2018stochastic, konow2021insights}, and the fundamental routes to instabilities provided here can provide some guidance on pursuing these in increasingly sophisticated systems even under constraints. Importantly, we note that the ideas used apply even with more constraints in the given system, as shown in the examples.


\bmhead{Acknowledgments}

E.~V-S. has received PhD funding from ANID, Beca Chile Doctorado en el extranjero, number 72210071.

The authors have no conflicts of interest to declare.

\bibliography{refs}

\newpage

\appendix

    \section{Useful results for designing unstable systems}\label{app:useful_results}
        Here, we collect a few additional lemmas that provide alternative ways of viewing the constructions illustrated above. In practice, particularly given additional constraints (i.e.~when one does not have complete freedom over $\D$ or $\jac$), these results can help one find parameters to exhibit Turing or wave instabilities. We begin by showing that \cref{lemma:positiveeig} has a converse, implying that, for a diagonalizable diffusion tensor, the existence of a principle unstable submatrix of $\m{Q}^{- 1} \jac \m{Q}$ is both necessary and sufficient for diffusion-driven instability.
        
        The authors in \cite[Theorem 1]{satnoianu} state that, when the diffusion matrix of a reaction-diffusion system is diagonal and the Jacobian matrix of the system in the absence of diffusion is S-stable, then the system does not admit Turing instabilities. In the proof of \cite[Theorem 2]{villar-sepulveda}, the authors noted that the original proof also applies to wave instabilities. We can directly use this result to obtain an equivalent statement in the case of reaction-cross-diffusion systems, as in the case of a non-diagonal $\D$.
        
        The growth rates (eigenvalues) of the original system given in \cref{generalsystem} are equivalent to those of $\m{Q}^{-1}\jac\m{Q}$. Hence by directly applying  \cite[Theorem 2]{villar-sepulveda} to this system, we have the following Lemma:
        \begin{lemma}\label{lemma:limitations}
            Assuming that $\D = \m{Q} \Db \m{Q}^{- 1}$ is diagonalizable with non-negative real eigenvalues, and $\jac$ is a stable matrix, then if $\m{Q}^{- 1} \jac \m{Q}$ is an S-stable matrix, the system \cref{generalsystem} does not admit Turing or wave instabilities.
        \end{lemma}
        
        The proofs of \cref{sub:Turing-1} and \cref{sub:proof-wave} both relied on the use of \cref{lemma:positiveeig}, where some diffusion coefficients were taken to be zero in an asymptotic limit. We can instead use an alternative approach where some diffusion coefficients tend to infinity instead. We can obtain this as a corollary of \cite[Lemma 6]{villar-sepulveda} applied to the diagonalized system $\m{Q}^{-1}\jac\m{Q}$.
        \begin{lemma} \label{lemma:nonpositiveeig}
            Let $\D$ be a diagonalizable matrix with non-negative real eigenvalues. Assume that $\D$ has $\ell$ eigenvalues $\sigma_{i_1}, \ldots, \sigma_{i_\ell}$ that tend to $\infty$, while all the others remain positive but finite. Then, there exist $\ell$ eigenvalues of $\linop$ having real parts that tend to $- \infty$ for every $\mu > 0$, while $n - \ell$ eigenvalues of said matrix tend to the eigenvalues of the principal submatrix of $\m{Q}^{- 1} \jac \m{Q}$ formed out of the rows and columns that are complementary to $i_1, \ldots, i_\ell$, as $\mu \to 0^+$.
        \end{lemma}

        Next, we develop two lemmas useful in dealing with systems that have constraints (i.e.~ones where we cannot control every element of $\m{D}$ or $\jac$). By an elementary formula (see, e.g.~\cite[Equation (1)]{Marcus1990}), we have the following result.
        \begin{lemma} \label{lemma:charpol}
            Let $\lambda\in \mathbb R$. Then
            \begin{align}
                \det\left(\jac - \lambda \m{I}_n\right) &= b_n \lambda^n + b_{n - 1} \lambda^{n - 1} + b_{n - 2} \lambda^{n - 2} + \ldots b_1 \lambda + b_0, \label{eq:charpol}
            \end{align}
            where
            \begin{align*}
                b_m &= (-1)^m \left(\sum_{1\leq i_1 < \ldots < i_{n-m} \leq n} M_{i_1 \ldots i_{n - m}}\right), \quad \text{for } 1 \leq m < n,
            \end{align*}
            $b_0 = \det(\jac)$, and $b_n = (- 1)^n$, where $M_{i_1 \ldots i_{n - m}}$ represents the determinant of the principal submatrix of $\jac$ composed of its rows and columns $i_1, \ldots, i_{n - m}$.
        \end{lemma}        
        \cref{lemma:charpol} enables us to prove a general result that can be used to generate Turing and wave instabilities in systems with more constraints. To motivate the coming result, we note that although \cref{remark_JNF} helped us prove the general results stated in \cref{th:main}, we can provide a generalization of this idea that let us embed larger unstable principal submatrices within the Jacobian matrix of the system via similarity. Furthermore, as we have already seen, there are cases in which different interactions between Turing and wave instabilities can occur, and being able to directly design such instances can be done using this generalization. For example, consider the Jacobian matrix given by
        \begin{align}
            \jac = \left(\begin{array}{ccc}
                 - 21.3 & - 78.2 & 87
                 \\
                 47.4 & - 1 & 1
                 \\
                 - 35 & - 3.1 & 3
            \end{array}\right), \label{linearizationextra}
        \end{align}
        which has eigenvalues $- 9.50278 \pm 81.3221 i, - 0.29445$, meaning it is stable. We note that it has a $2 \times 2$ unstable principal submatrix at the bottom-right corner with eigenvalues $0.0513167$ and $1.94868$, with an element within this submatrix being 3. In particular, if we consider a diffusion matrix given by
        \begin{align}
            \D = \begin{pmatrix}
                100 & 0 & 0
                \\
                0 & 0.0001 & 0
                \\
                0 & 0 & \delta
            \end{pmatrix}, \label{diffmatrixextra}
        \end{align}
        the dispersion relation of a potential system we can generate with these matrices for different values of $\delta$ that range from 0 to 0.002 in uniform steps of size 0.0004 is shown in \cref{fig:disprelextra}. Here, we note that there is a strong interaction between Turing and wave instabilities. In fact, when $\delta = 0$, the largest eigenvalue converges to a real number as $k \to \infty$, as it is expected. However, for larger values of $\delta$, there is an interaction between real and complex-conjugate eigenvalues as there is a range of values of said parameter under which two complex conjugate eigenvalues split into two real eigenvalues for a range of values of $k$. This implies that several different scenarios can be generated through different unstable submatrices in the Jacobian matrix of the system. The question is then how can we generate a stable matrix from an unstable one? Part of the answer is provided in the following lemma, which once again, allows for several degrees of freedom.
        \begin{figure}
            \centering
            \includegraphics[width = 0.8 \linewidth]{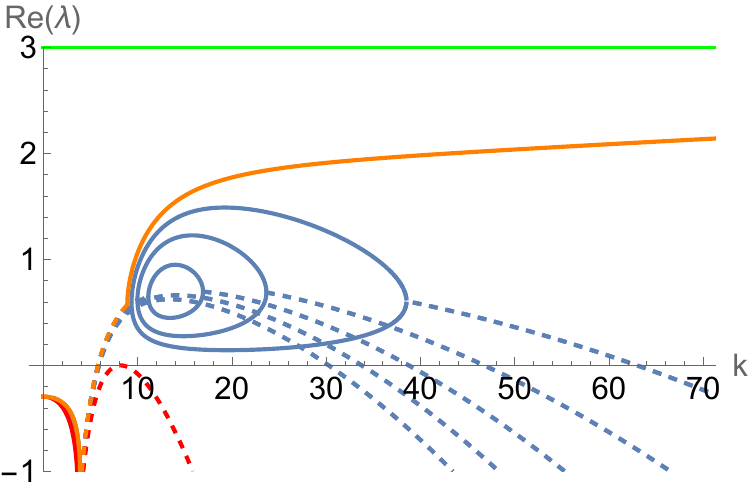}
            \caption{Dispersion relation of a system whose linearization is given by \cref{linearizationextra} and the diffusion matrix is set as \cref{diffmatrixextra}.}
            \label{fig:disprelextra}
        \end{figure}
        
        To introduce the following result, let $\m{K}$ be a square matrix of size $\ell \times \ell$. Furthermore, let $\m{S_{i_1, \ldots, i_m}}$ be its principal submatrix comprising only its $i_1, \ldots, i_m$ rows and columns and $M_{i_1, \ldots, i_m}$ its determinant. Furthermore, for $p, q \in \{1, \ldots , m\}$, we denote by $\m{S_{i_1, \ldots, i_m}^{i_p, i_q}}$ the submatrix of $\m{S_{i_1, \ldots, i_m}}$ after extracting its $i_p$-th row and $i_q$-th column, and $M_{i_1, \ldots, i_m}^{i_p, i_q}$ again denotes the corresponding determinant. We then have the following result.
        \begin{lemma} \label{lemma:construction}
            Let $\ell \geq 1$ be an integer and let $\m{A}$ be a matrix of size $(\ell + 1) \times (\ell + 1)$ with $\ell + 1$ different eigenvalues. Let $\m{K}$ be any matrix of size $\ell \times \ell$ such that there exists a set of scalars $\left\{c_{s, 1}\right\}_{s = 2}^{\ell + 1}$ such that the set of vectors $\left\{\mbf v_1, \ldots , \mbf v_{\ell + 1}\right\}$ is linearly independent, where these vectors are defined by coordinates as
            \begin{align*}
                v_{m, 1} &= \sum_{1 \leq i_1 < \ldots < i_{\ell - m - 1} \leq \ell} M_{i_1 \ldots i_{\ell - m - 1}},
                \\
                v_{m, r + 1} &= \sum_{\substack{1 \leq i_1 < \ldots < i_{\ell - m - 1} \leq \ell \\  s \in \left\{i_1, \ldots, i_{\ell - m - 1}\right\}}} c_{s, 1} M_{i_1 \ldots i_{\ell - m - 1}}^{s, i_r}, \quad \text{for } 1 \leq r \leq \ell, \quad i_r \in \left\{i_1, \ldots, i_{\ell - m - 1}\right\},
            \end{align*}
            with $\mbf v_m = \left(v_{m, 1}, \ldots, v_{m, \ell + 1}\right)^\intercal$ for each $1 \leq m \leq \ell + 1$, then there exists a set of real numbers $\left\{c_{1, s}\right\}_{s = 1}^{\ell + 1}$ such that the matrix
            \begin{align*}
                \m{B} := \left(\begin{array}{c|ccc}
                    c_{1, 1} & c_{1, 2} & \cdots & c_{1, \ell + 1}
                    \\
                    \hline
                    c_{2, 1} & & & 
                    \\
                    \vdots & & \m{K}
                    \\
                    c_{\ell + 1, 1}
                \end{array}\right)
            \end{align*}
            is similar to $\m{A}$.
        \end{lemma}
    
        \begin{proof}
            As the eigenvalues of $\m{A}$ are all different, then a matrix is similar to $\m{A}$ if and only if it has the same characteristic polynomial as $\m{A}$.
            By \cref{lemma:charpol}, the characteristic polynomial of $B$ is given by
            \begin{align*}
                b_{\ell + 1}  \lambda^{\ell + 1} + b_\ell \lambda^\ell + \ldots + b_1 \lambda + b_0,
            \end{align*}
            where
            \begin{align}
                b_m &=
                (- 1)^m \left(c_{1, 1} \sum_{1 \leq i_1 < \ldots < i_{\ell - m - 1} \leq \ell} M_{i_1 \ldots i_{\ell - m - 1}}\right.
                \\
                & \quad + \sum_{\substack{1 \leq i_1 < \ldots < i_{\ell - m - 1} \leq \ell \\ r \in \left\{i_1, \ldots, i_{\ell - m - 1}\right\}}} c_{1, r} \left(\sum_{s \in \left\{i_1, \ldots, i_{\ell - m - 1}\right\}} c_{s, 1} M_{i_1 \ldots i_{\ell - m - 1}}^{s, r}\right) \notag
                \\
                & \qquad \left. + \sum_{1 \leq i_1 < \ldots < i_{\ell - m} \leq \ell} M_{i_1 \ldots i_{\ell - m}}\right), \quad \text{for } 1 \leq m \leq \ell, \label{b_m}
            \end{align}
            $b_0 = \det\left(\m{B}\right)$, and $b_{\ell + 1} = (- 1)^{\ell + 1}$. Observe that the first two sums in \cref{b_m} are linear combinations of the coefficients $\left\{c_{1, s}\right\}_{s = 1}^{\ell + 1}$ whilst the third one is a constant term depending only on $\m{K}$.
            
            On the other hand, the characteristic polynomial of $\m{A}$ has the same form with different coefficients:
            \begin{align*}
                a_{\ell + 1} \lambda^{\ell + 1} + a_\ell \lambda^\ell + \ldots + a_1 \lambda + a_0,
            \end{align*}
            where $a_p \in \mathbb R$ for $p = 0, 1, \ldots, \ell + 1$.
            
            We need to solve the system of equations given by $a_m = b_m$, for each $m = 0, 1, \ldots, \ell$ (clearly, $a_{\ell + 1} = b_{\ell + 1}$).
            
            Thus, we need to solve a system of equations with the following form:
            \begin{align*}
                \begin{pmatrix}
                    p_{0, 1} & p_{0, 2} & \cdots & p_{0, \ell + 1}
                    \\
                    p_{1, 1} & p_{1, 2} & \cdots & p_{1, \ell + 1}
                    \\
                    \vdots & \vdots & \ddots & \vdots
                    \\
                    p_{\ell, 1} & p_{\ell, 2} & \cdots & p_{\ell, \ell + 1}
                \end{pmatrix} \begin{pmatrix}
                    c_{1, 1}
                    \\
                    c_{1, 2}
                    \\
                    \vdots
                    \\
                    c_{1, \ell + 1}
                \end{pmatrix} = \mbf b,
            \end{align*}
            where, for $0 \leq m \leq \ell$,
            \begin{align*}
                p_{m, 1} &=
                \sum_{1 \leq i_1 < \ldots < i_{\ell - m - 1} \leq \ell} M_{i_1 \ldots i_{\ell - m - 1}},
                \\
                p_{m, r} &=
                \sum_{\substack{1 \leq i_1 < \ldots < i_{\ell - m - 1} \leq \ell \\  s \in \left\{i_1, \ldots, i_{\ell - m - 1}\right\}}} c_{s, 1} M_{i_1 \ldots i_{\ell - m - 1}}^{s, i_r}, \quad \text{for } 2 \leq r \leq \ell + 1, i_r \in \left\{i_1, \ldots, i_{\ell - m - 1}\right\}.
            \end{align*}        
            Note that, by assumption, there exist coefficients $\left\{c_{s, 1}\right\}_{s = 2}^{\ell + 1}$ so that the matrix of coefficients of this system of equations is invertible, which implies the existence of a solution $\left\{c_{1, s}\right\}_{s = 1}^{\ell + 1}$ so that the conclusion holds.
        \end{proof}

    \section{Forms of similarity matrices}\label{app:R_similarity}
        Here, we record the form of the matrices $\m{R}$ in the proof of \cref{lemma:wave}. These all arise by solving the 9 equations given by $\m{R} \m{K} = \m{C_1} \m{R}$, $\m{R} \m{K} = \m{C_2} \m{R}$, or the 16 equations $\m{R} \m{K} = \m{C_3} \m{R}$, leaving the leftmost column as free parameters.
            \subsubsection*{Similarity matrix for $\m{K}$ and $\m{C_1}$:}
            \begin{align*}
                R = \begin{pmatrix}
                    r_{1, 1} & r_{1, 2} & r_{1, 3}
                    \\
                    r_{2, 1} & r_{2, 2} & r_{2, 3}
                    \\
                    r_{3, 1} & r_{3, 2} & r_{3, 3}
                \end{pmatrix},
            \end{align*}
            where
            \begin{align*}
                r_{1, 2} &= - \frac{1}{\alpha^2 + \gamma^2} \left(- 2 \alpha r_{1, 1} + \gamma z^2 r_{1, 1} - 2 \alpha z r_{1, 1} + 2 \gamma z r_{1, 1}\right),
                \\
                r_{1, 3} &= - \frac{1}{\alpha^2 + \gamma^2} \left(- 2 \gamma  r_{1, 1} - \alpha z^2 r_{1, 1} - 2 \alpha z r_{1, 1} - 2 \gamma z r_{1, 1}\right),
                \\
                r_{2, 2} &= - \frac{1}{\alpha^2 + \gamma^2} \left(- \alpha \ell r_{3, 1} + \gamma \ell r_{3, 1} + \gamma \ell y r_{3, 1} - 2 \alpha  r_{2, 1} - \alpha y r_{2, 1}\right.
                \\
                & \quad \left. + \gamma y r_{2, 1} + \gamma y z r_{2, 1} - \alpha z r_{2, 1} + \gamma z r_{2, 1}\right),
                \\
                r_{2, 3} &= - \frac{1}{\alpha^2 + \gamma^2} \left(- \alpha \ell r_{3, 1} - \gamma \ell r_{3, 1} - \alpha \ell y r_{3, 1} - 2 \gamma  r_{2, 1} - \alpha y r_{2, 1}\right.
                \\
                & \quad \left. - \gamma y r_{2, 1} - \alpha y z r_{2, 1} - \alpha z r_{2, 1} - \gamma z r_{2, 1}\right),
                \\
                r_{3, 2} &= - \frac{1}{\alpha^2 + \gamma^2} \left(- 2 \alpha r_{3, 1} - \alpha y r_{3, 1} + \gamma y r_{3, 1} + \gamma y z r_{3, 1} -\alpha z r_{3, 1} + \gamma z r_{3, 1}\right),
                \\
                r_{3, 3} &= - \frac{1}{\alpha^2 + \gamma^2} \left(- 2 \gamma r_{3, 1} - \alpha y r_{3, 1} - \gamma y r_{3, 1} - \alpha y z r_{3, 1} -\alpha z r_{3, 1} - \gamma z r_{3, 1}\right).
            \end{align*}
            \subsubsection*{Similarity matrix for $\m{K}$ and $\m{C_2}$:}
            \begin{align*}
                R = \begin{pmatrix}
                    r_{1, 1} & r_{1, 2} & r_{1, 3}
                    \\
                    r_{2, 1} & r_{2, 2} & r_{2, 3}
                    \\
                    r_{3, 1} & r_{3, 2} & r_{3, 3}
                \end{pmatrix},
            \end{align*}
            where
            \begin{align*}
                r_{1, 2} &= - \frac{1}{\alpha^2 + \gamma^2} \left(- \alpha \ell_2 r_{2, 1} + \gamma \ell_2 r_{2, 1} + \gamma \ell_1 \ell_2 r_{3, 1} + \gamma \ell_2 z r_{2, 1}\right.
                \\
                & \quad \left. - 2 \alpha  r_{1, 1} + \gamma z^2 r_{1, 1} - 2 \alpha z r_{1, 1} + 2 \gamma z r_{1, 1}\right),
                \\
                r_{1, 3} &= - \frac{1}{\alpha^2 + \gamma^2} \left(- \alpha \ell_2 r_{2, 1} -\alpha \ell_1 \ell_2 r_{3, 1} - \gamma \ell_2 r_{2, 1} - \alpha \ell_2 z r_{2, 1}\right.
                \\
                & \quad \left. - 2 \gamma r_{1, 1} - \alpha z^2 r_{1, 1} - 2 \alpha z r_{1, 1} - 2 \gamma z r_{1, 1}\right),
                \\
                r_{2, 2} &= - \frac{1}{\alpha^2 + \gamma^2} \left(- \alpha \ell_1 r_{3, 1} + \gamma \ell_1 r_{3, 1} + \gamma \ell_1 z r_{3, 1}\right.
                \\
                & \quad \left. - 2 \alpha r_{2, 1} + \gamma z^2 r_{2, 1} - 2 \alpha  z r_{2, 1} + 2 \gamma z r_{2, 1}\right),
                \\
                r_{2, 3} &= - \frac{1}{\alpha^2 + \gamma^2} \left(- \alpha \ell_1 r_{3, 1} - \gamma \ell_1 r_{3, 1} - \alpha \ell_1 z r_{3, 1}\right.
                \\
                & \quad \left. - 2 \gamma r_{2, 1} - \alpha z^2 r_{2, 1} - 2 \alpha z r_{2, 1} - 2 \gamma z r_{2, 1}\right),
                \\
                r_{3, 2} &= - \frac{1}{\alpha^2 + \gamma^2} \left(- 2 \alpha r_{3, 1} + \gamma z^2 r_{3, 1} - 2 \alpha z r_{3, 1} + 2 \gamma z r_{3, 1}\right),
                \\
                r_{3, 3} &= - \frac{1}{\alpha^2 + \gamma^2} \left(- 2 \gamma r_{3, 1} - \alpha z^2 r_{3, 1} - 2 \alpha z r_{3, 1} - 2 \gamma z r_{3, 1}\right).
            \end{align*}

        We omit the form of $\m{R}$ in the case of similarity between $\m{K}$ and $\m{C_3}$, but remark that each term in it has a denominator of the form $5 \alpha^2 \left(\delta^2 + \varepsilon^2\right)$, and hence $\m{R}$ is well-defined for $\alpha \neq 0$ and $\delta^2 + \varepsilon^2 \neq 0$.
        
        While the precise forms of these similarity matrices arise in applying the result given in \cref{th:main}, their constructions are a bit tedious and substantially easier to obtain using a computer algebra system. Importantly, we only need such constructions to arrive at the result given in \cref{lemma:wave}, and note that this can be satisfied using other constructions as well.

\end{document}